\documentclass[10pt,journal,compsoc]{IEEEtran}
\usepackage{mathtools}
\usepackage{empheq}
\mathtoolsset{showonlyrefs}
\usepackage{algpseudocode, algorithm}
\usepackage{indentfirst}
\usepackage{amsthm}
\usepackage{amsfonts}
\usepackage{multirow}
\usepackage{color}
\usepackage{subfig}

\DeclareMathOperator*{\argmax}{argmax}

\newtheorem{theorem}{Theorem}
\newtheorem{definition}{Definition}
\newtheorem{remark}{Remark}

\ifCLASSOPTIONcompsoc
   \usepackage[nocompress]{cite}
\else
   \usepackage{cite}
\fi
\ifCLASSINFOpdf
\else
\fi
\usepackage{url}

\begin{document}
%
\title{Decentralized Computation Offloading and Resource Allocation in Heterogeneous Networks with Mobile Edge Computing}

%
%
%
%


\author{Quoc-Viet Pham, Tuan LeAnh, Nguyen H. Tran, and~Choong Seon Hong
\IEEEcompsocitemizethanks{\IEEEcompsocthanksitem The authors are all with the Department of Computer Science and Engineering, Kyung Hee University, Republic of Korea. Dr. Choong Seon Hong is the corresponding author.\protect\\
E-mail:  \{vietpq90, latuan, nguyenth, cshong\}@khu.ac.kr}
}

\markboth{IEEE Transactions on Mobile Computing}{Q.V. PHAM \MakeLowercase{\textit{et al.}}: Decentralized Computation Offloading and Resource Allocation in HetNets with MEC}


\IEEEtitleabstractindextext{%
\begin{abstract}
We consider a heterogeneous network with mobile edge computing, where a user can offload its computation to one among multiple servers. In particular, we minimize the system-wide computation overhead by jointly optimizing the individual computation decisions, transmit power of the users, and computation resource at the servers. The crux of the problem lies in the combinatorial nature of multi-user offloading decisions, the complexity of the optimization objective, and the existence of inter-cell interference. Then, we decompose the underlying problem into two subproblems: i) the offloading decision, which includes two phases of user association and subchannel assignment, and ii) joint resource allocation, which can be further decomposed into the problems of transmit power and computation resource allocation. To enable distributed computation offloading, we sequentially apply a many-to-one matching game for user association and a one-to-one matching game for subchannel assignment. Moreover, the transmit power of offloading users is found using a bisection method with approximate inter-cell interference, and the computation resources allocated to offloading users is achieved via the duality approach. The proposed algorithm is shown to converge and is stable. Finally, we provide simulations to validate the performance of the proposed algorithm as well as comparisons with the existing frameworks.
\end{abstract}

}

\maketitle
\IEEEdisplaynontitleabstractindextext
\IEEEpeerreviewmaketitle

\section{Introduction}	
\label{Sec:Introduction}
\IEEEPARstart{W}{ith} the radically increasing popularity of mobile terminals such as smart phones and tablet computers, a wide-range of mobile applications are constantly emerging, including real-time online gaming, augmented reality, natural language processing, and ultra-high-definition video streaming. These new mobile applications usually have stringent requirements of real-time communication, high energy efficiency, and intensive computation. However, mobile devices are often constrained with limited battery capacity and computation capability. To tackle these issues, mobile cloud computing (MCC) has been successfully developed for the last decade \cite{Chen2015Decentralized}. In MCC, a mobile user exploits remote cloud data centers, which are enormously powerful in terms of computation and storage resources, by offloading its computation tasks and data through the core of the wireless networks \cite{Fernando2013Mobile}. The achievable advantages of MCC include extension of the battery life time, provision of a high storage pool for mobile users, and the ability to deploy new sophisticated applications in mobile devices \cite{Mach2017Mobile}. However, there are serious limitations of MCC including high latency, low scalability, and high burden on fronthaul links. To address the drawbacks of MCC, a new trend called mobile edge computing (MEC) has been proposed that moves the cloud services and functions to the edge of the mobile networks. In this paper, we consider an MEC system with \textit{multiple MEC servers} and investigate an efficient scheme of \textit{distributed computation offloading} and resource (\textit{computation resource} and \textit{communication resource}) allocation.

\subsection{Prior Work and Motivation}
Different aspects of MEC systems have been thoroughly reviewed in surveys \cite{Mach2017Mobile, Mao2017_aSurveyMEC, Abbas2018Mobile}. In \cite{Mach2017Mobile}, the authors first reviewed different MEC concepts for the integration of cloud functionalities at the edge of the mobile networks, e.g., small cell clouds, mobile micro clouds, and fast moving personal clouds. Then, computation offloading in MEC systems was reviewed from the viewpoint of offloading decisions, full offloading and partial offloading. With this track, the authors in \cite{Mao2017_aSurveyMEC} presented a survey of MEC systems from the perspective of wireless communications. The literature \cite{Abbas2018Mobile} focused on a survey of emerging application scenarios and privacy and security issues in MEC systems. From the above surveys, computation offloading is a major part of any MEC system. A computation offloading scheme is generally used to decide whether mobile users should offload their computation tasks to the MEC servers or not. Moreover, a computation offloading scheme depends on many factors such as application models and requirements, computation capabilities of the mobile users and remote MEC server, radio resources, and backhaul capacity. For example, if computation tasks and data of a mobile application are allowed to be partitioned/parallelized to different parts, partial computation offloading to multiple MEC servers are available. 

Recent years have seen a large number of research literature on computation offloading in MEC systems. From the perspective of a single user, computation offloading has been considered in \cite{Zhang2013Energy, Wang2016Mobile, Dinh2017Offloading}. Considering a cloud computing model, in order to preserve energy of mobile devices, Zhang \textit{et al.} in \cite{Zhang2013Energy} proposed dynamically adjusting the CPU frequency for mobile execution and scheduling the data transmission rate for cloud execution. In addition, the authors derived an optimal threshold policy for computation offloading decisions i.e., mobile execution or cloud execution. The studies \cite{Wang2016Mobile, Dinh2017Offloading} implemented the dynamic voltage scaling technique with computation offloading for different objectives. Specifically, the authors in \cite{Wang2016Mobile} optimized the operating frequency and transmit power of the mobile devices as well as an offloading ratio for local and remote computing, in a single-task MEC system. Unlike \cite{Wang2016Mobile}, Dinh \textit{et al.} in \cite{Dinh2017Offloading} considered an MEC system with multiple access points (APs), where the computation tasks of mobile users are independent and each one can be executed either remotely by an AP or locally by the mobile device.

A number of studies have also been devoted to computation offloading with multiple users in MEC systems \cite{Wang2017Joint, Wang2017Computation, Bi2017Computation, Chen2016Efficient, CWang2017Joint, CWang2017Computation}. For instance, results on the integration of wireless power transfer and mobile edge computing were developed in \cite{Wang2017Joint, Wang2017Computation, Bi2017Computation}. Wang \textit{et al.} in \cite{Wang2017Joint} considered a time-division multiple access based MEC system where a multi-antenna access point transmits energy beamforming to charge multiple users, and the formulated a method to minimize the energy consumption at the access point. Similarly, the weighted sum of computation rate maximization problem was considered in \cite{Wang2017Computation} and solved by optimizing the transmit beamforming of the access point, the computation task partition for offloading and local computing, and the time allocation among users, which was accomplished using the Lagrange dual technique. As opposed to \cite{Wang2017Joint, Wang2017Computation} where partial offloading was considered, a computation rate maximization problem with binary offloading was formulated in \cite{Bi2017Computation}, and then solved by either a decomposition technique using the coordinated descent method or a joint algorithm using the alternating direction method of multipliers (ADMM) approach. The authors in \cite{Chen2016Efficient} first showed that finding the maximum number of offloading users is NP-Hard, and adopted a game theoretic approach to find the computation offloading decision in a distributed manner. In \cite{CWang2017Joint}, since both computation offloading and interference management are interdependent and jointly affect the network performance, a framework of computation offloading and interference management in heterogeneous networks (HetNets) was considered.  The authors in \cite{CWang2017Computation} studied a framework of computation offloading, resource allocation, and caching in HetNets, which consisted of two steps: convexifying the original problem and applying the ADMM method to propose a distributed algorithm. 

Notwithstanding numerous studies on computation offloading and resource allocation in multi-user MEC systems, these works generally considered only one MEC server. Since (ultra-dense) HetNets have been considered as important parts of 5G networks, it seems quite possible that there are multiple MEC servers (each one is connected to and collocated with a small base station) over a specific area to provide connectivity and services to multiple users. Moreover, due to the dynamics and unplanned deployment of HetNets, as well as the possibility of missing a central entity, it is necessary to design distributed computation offloading approaches. There are few existing studies on computation offloading in multi-cell heterogeneous networks with mobile edge computing \cite{Sardellitti2015Joint, Wang2016Mobile, Tuyen2017Joint, Sato2017Radio, Emara2017MEC}. Considering a HetNet where multiple SeNBs connect to a common cloud server and assuming that the sets of offloading users and non-offloading users are given, the authors in \cite{Sardellitti2015Joint} jointly optimized the transmit precoding of users and computation resources of the cloud server so as to minimize the total energy consumption of mobile users and guarantee the latency constraints. Besides single cloud servers, the authors in \cite{Wang2016Mobile} also extended their work to multiple cloud servers; however, the problem was considered with only one mobile device. Additionally, a partial computation offloading policy was assumed and the computation resources at the MEC servers were not considered. In a single cloud server, Lyu \textit{et al.} in \cite{Lyu2017Multiuser} found the optimal offloading decision by a heuristic algorithm. The work in \cite{Lyu2017Multiuser} was extended to multiple cloud servers in \cite{Tuyen2017Joint}, where the offloading decision problem was solved by heuristically performing either a remove operation or an exchange operation at each step. Sato and Fujii in \cite{Sato2017Radio} proposed two approaches for computation offloading; the first and second use the radio environment map to predict connectivity and the received signal power to estimate distance between the offloading user and MEC servers, respectively. However, only one user is considered and user association with the MEC servers is based on the conventional concept, i.e., reference signal received power (RSRP) based user association. The authors in \cite{Emara2017MEC} considered a multi-tier HetNet, where an MEC server is placed at each tier, illustrating that the proposed user association method is superior to the conventional user association scheme. Nevertheless, in both \cite{Sato2017Radio} and \cite{Emara2017MEC}, the offloading decision, uplink transmit power of the mobile users, and computation resources at the MEC servers are not taken into consideration

\subsection{Contributions of this paper}
We aim to solve the \textit{computation offloading decision} problem in a \textit{distributed} manner and efficiently optimize the \textit{radio and computation resource allocation} in \textit{multi-cell heterogeneous networks} with mobile edge computing. It has been discussed in \cite{Kamel2016Ultra} that the concept of (ultra-dense) HetNets illustrates a new paradigm shift in next-generation networks, where a large number of small cells are deployed in the hotspots. With the concept of (ultra-dense) HetNets and recent advancements in computing hardware, there may exist multiple MEC servers, where each one is connected to and collocated with an SeNB and is able to execute multiple computation tasks. In HetNets with mobile edge computing, a mobile user can either handle its computation locally or send a request to one among multiple MEC servers for computation offloading. In particular, we are interested in minimizing the system-wide computation overhead by jointly optimizing the individual \textit{computation decision}, \textit{transmit power} of mobile users, and \textit{computation resources} for offloading users at the MEC servers. The considered problem represents difficulties caused by the combinatorial nature of multi-user computation offloading decisions, the complexity of the optimization objective, and the existence of inter-cell interference among offloading users. While most existing studies for computation offloading focus on either centralized heuristic algorithms \cite{Lyu2017Multiuser, Tuyen2017Joint} or decentralized approaches for computation offloading in single-server systems using game-theoretic concepts \cite{Chen2015Decentralized, Chen2016Efficient}, they are not applicable to (ultra-dense) HetNets with mobile edge computing. In this paper, matching theory, a powerful tool to design distributed algorithms for a large number of resource allocation problems in wireless communication, including heterogeneous networks, device-to-device (D2D) communications, cognitive radio networks, and physical layer security \cite{Gu2015Matching, Bayat2016Matching, Han2017Matching}, is adopted to provide the distributed computation offloading decision in multi-user multi-server HetNets. Our contributions can be summarized, as follows:
\begin{itemize}
\item In terms of the system model and problem formulation, we consider a network scenario with multiple SeNBs collocated with the corresponding MEC servers and multiple users, and define the objective function as the system-wide computation overhead. Then, an optimization problem is formulated subject to constraints on the MEC server and subchannel selections, maximum transmit power of mobile devices, and maximum computation resources at the MEC servers. After that, the underlying problem is decomposed into two independent parts: 1) the computation offloading decision problem, which includes two phases of user association and subchannel assignment, and 2) resource allocation, which can be further decomposed into the transmit power of mobile users and computation resource allocation at the MEC servers. 

\item In terms of the mathematical framework, we adopt two matching games to design algorithms for user association and subchannel assignment. Accordingly, a decentralized approach is investigated to determine the offloading decision. With the proposed computation offloading scheme, 1) users decide to offload their computation tasks if and only if computation offloading is advantageous to the offloading users and 2) mobile users and MEC servers make the offloading decision in a distributed and autonomous fashion. In addition, we approximate the inter-cell interference and find the transmit power of mobile users using a bisection method, and then solve the computation resource allocation problem via the Lagrange dual approach. 

\item In terms of the performance evaluation, we validate the performance of the proposed algorithm through extensive numerical experiments. Furthermore, we compare our proposed algorithm with four existing solutions: local computing only, offloading only, the heuristic offloading decision algorithm (HODA) proposed in \cite{Lyu2017Multiuser}, and the heuristic joint task offloading and resource allocation (hJTORA) proposed in \cite{Tuyen2017Joint}. The results illustrate that our proposed algorithm can achieve a performance improvement from computation offloading in terms of the number of offloading users and the system-wide computation overhead.
\end{itemize}

Our paper is organized as follows. The system model and optimization problem are explained in Section~\ref{Sec:System_Model}. In Section~\ref{Sec:Algorithm}, we apply the decomposition technique to decompose the underlying problem into subproblems and propose efficient methods to solve the problems. Simulation results are presented and discussed in Section~\ref{Sec:Simulation}. Finally, conclusions and future works are drawn in Section~\ref{Sec:Conclusion}.

\section{System Model and Problem Formulation}	
\label{Sec:System_Model}
\subsection{Network Model}
We consider a multi-cell MEC system as illustrated in Fig.~\ref{Fig:SystemModel_Example}. In the considered network, each MEC server is assumed to be collocated with an SeNB and each cell can be a small cell such as a femtocell or a picocell\footnote{From that point, we use ``MEC server" when referring to the related concepts of computation resource and computation offloading, while using ``BS" or ``SeNB" when mentioning interference, radio resource, and user association.}. In general, each computing server deployed by the network operator has a moderate computing capability and has wireless channel connections to mobile users through the corresponding BS. Small cells operate in an overlaid manner, i.e., each small cell is able to reuse the whole spectrum of the macro cell and interference among small cells exists. In addition, the spectrum in a cell is divided into subchannels and is orthogonally assigned to mobile users, and thus intra-cell interference can be fully mitigated. To enable tractable analysis and obtain useful insights, we employ a quasi-static network scenario where mobile users remain unchanged during the computation offloading period while they change across different periods. The general scenario, where users leave or join dynamically during the offloading period, is not our focus in this paper and can be considered as future work. 

\begin{figure}[!t]
\centering
\includegraphics[width=0.80\linewidth]{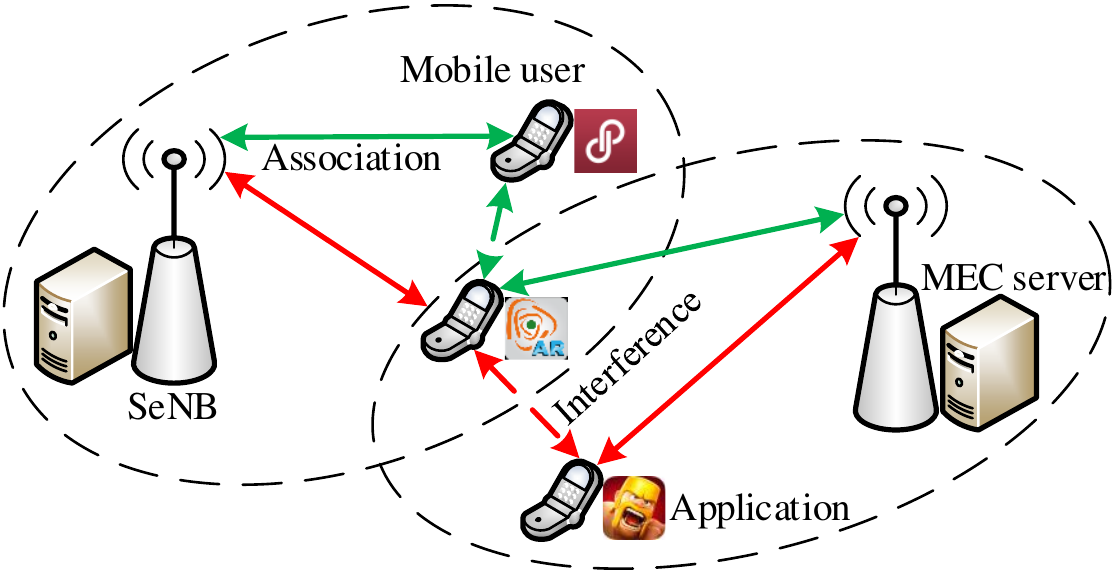}
\caption{An MEC system with two SeNBs, each with an MEC server, and three users, each with a distinct application.}
\label{Fig:SystemModel_Example}
\end{figure}

Denote by $\mathcal{M} = \left\lbrace 1,2,...,M \right\rbrace$ the set of SeNBs and by $m$ the index of the  $m$th SeNB. The set of mobile users is denoted by $\mathcal{N}= \left\lbrace 1,2,...,N \right\rbrace$ and $n$ is referred to as the $n$th mobile user. In this paper, we use Orthogonal Frequency-Division Multiple Access (OFDMA) as the multiple access scheme in the uplink. Assume that there are $S$ subchannels in a small cell, then the set of subchannels is denoted $\mathcal{S} = \left\lbrace 1,2,...,S \right\rbrace$ and $s$ is used to refer to the $s$th subchannel. Each mobile user is assigned to at most one subchannel and a subchannel is assigned to at most one mobile user. The general case that a mobile user is assigned to multiple subchannels and a subchannel is assigned to multiple mobile users will be considered in future work.

\subsection{Communication Model}
For every small cell, an MEC server is collocated with the corresponding SeNB; therefore, a mobile user can offload its computation task to the MEC server via the SeNB. We define the offloading decision profile as $\boldsymbol{A} = \left\lbrace a_{nm}^{s} | n \in \mathcal{N}, m \in \mathcal{M}, s \in \mathcal{S} \right\rbrace$. Specifically, $a_{nm}^{s} = 1$ if the user $n$ offloads its computation task to the MEC server $m$ on the subchannel $s$, and $a_{nm}^{s} = 0$ otherwise. Since each computation task can be either computed locally or remotely, we have the following constraint:
\begin{equation} \label{Eq:Constr_UE_OffloadingDecision}
\sum\nolimits_{m \in \mathcal{M}}\sum\nolimits_{s \in \mathcal{S}}a_{nm}^{s} \leq 1, \forall n \in \mathcal{N}.
\end{equation}
Each SeNB assigns a subchannel to at most one mobile user, so the following constraint must be satisfied:
\begin{equation} \label{Eq:Constr_MEC_OffloadingDecision}
\sum\nolimits_{n \in \mathcal{N}}a_{nm}^{s} \leq 1, \forall m \in \mathcal{M}, s \in \mathcal{S}.
\end{equation}
In addition, since each MEC server is collocated with a SeNB, which often has limited hardware capability \cite{Chandrasekhar2008Femtocell}, the number of mobile users offloading to an MEC server should be constrained by 
\begin{equation} \label{Eq:Constr_Quota}
\sum\nolimits_{n \in \mathcal{N}}\sum\nolimits_{s \in \mathcal{S}}a_{nm}^{s} \leq q_{m}, \forall m \in \mathcal{M},
\end{equation}
where $q_{m}$ is called a \textit{quota}, which represents the maximum number of mobile users the MEC server $m$ can serve. In this paper, the quota $ q_{m} $ corresponds equally to the number of subchannels in each cell.

Given the offloading decision profile $\boldsymbol{A}$, the uplink data rate of the mobile user $n$ when it offloads the computation task $I_{n}$ to the MEC server $m$ over the subchannel $s$ is defined as $ R_{nm}^{s}\left( \boldsymbol{A}^{s}, \boldsymbol{P} \right) = B_{s} \log_{2}\left( 1 + \Gamma_{nm}^{s}\left( \boldsymbol{A}^{s}, \boldsymbol{P} \right) \right), \forall n, m, s, $ where $B_{s}$ is the bandwidth of the subchannel $s$, $\boldsymbol{A}^{s} = \left\lbrace a_{nm}^{s} | n \in \mathcal{N}, m \in \mathcal{M} \right\rbrace$ is the offloading decision profile on the subchannel $s$, and $\Gamma_{nm}^{s}$ is the signal-to-interference-plus-noise ratio (SINR) of the mobile user $n$ that offloads to the MEC server $m$ on the subchannel $s$, which can be written as 
\begin{equation} \label{Eq:SINR}
\Gamma_{nm}^{s}\left( \boldsymbol{A}^{s}, \boldsymbol{P} \right) = \frac{p_{n}^{s}h_{nm}^{s}}{n_{0} + \sum\limits_{j \neq m}\sum\limits_{k \in \mathcal{N}_{j}}a_{kj}^{s}p_{k}h_{km}^{s}}.
\end{equation}
Here, $n_{0}$ is the power spectral density of additive white Gaussian noise which is identical for all mobile users, $\boldsymbol{P} = \left\lbrace \boldsymbol{p}_{1},..., \boldsymbol{p}_{n},...,\boldsymbol{p}_{N} \right\rbrace$ is the transmit power vector of all mobile users, $\boldsymbol{p}_{n} = \lbrace p_{n}^{1},...,p_{n}^{S} \rbrace$ is the transmit power vector of mobile user $n$ with $p_{n}^{s}$ being the transmit power (in Watts) of mobile user $n$ on subchannel $s$, and $h_{nm}^{s}$ is the uplink channel gain from the mobile user $n$ to the SeNB $m$ on the subchannel $s$. The second term of the denominator in~\eqref{Eq:SINR} is the total interference from other mobile users offloading to other MEC servers on the same subchannel $s$. Correspondingly, the data rate of mobile user $n$ with the SeNB $m$ is given by $ R_{nm}\left( \boldsymbol{A}, \boldsymbol{P} \right) = \sum\nolimits_{s \in \mathcal{S}}a_{nm}^{s}R_{nm}^{s}\left( \boldsymbol{A}^{s}, \boldsymbol{P} \right), \forall n \in \mathcal{N}, m\in \mathcal{M}. $

\subsection{Computation Model of Mobile Devices}
Each mobile user $n$ has a computation task $I_{n} = \left\lbrace \alpha_{n}, \beta_{n}, \omega_{n} \right\rbrace$ \cite{Chen2016Efficient, Mao2017_aSurveyMEC}, where $\alpha_{n}$ is the computation input data size (in bits), $\beta_{n}$ is the number of CPU cycles required to complete the task $I_{n}$, i.e., computation workload or computation intensity, and $\omega_{n}$ is the computational result, i.e., output data (in bits). Each computation task can be executed either locally or remotely\footnote{Generally, there are two types of computation offloading: binary offloading and partial offloading. In the former case, as considered in our work, an integrated or a simple task can not be partitioned into sub-tasks and then must be executed either locally  at the mobile user or remotely at the MEC server. In the meanwhile, in partial offloading a task can be arbitrarily divided into sub-tasks, which can be executed at multiple MEC servers \cite{Mao2017_aSurveyMEC}.}. 

For local computing, the computation task $I_{n}$ is executed by the mobile user $n$. We denote $f_{n}^{l}$ as the computational capability (in CPU cycles per second) of the mobile user $n$, where the superscript $ l $ stands for \textit{local}. Due to the heterogeneity of the mobile devices, different mobile users can have different computational capabilities. Let $t_{n}^{l}$ be the completion time of the task $I_{n}$ by the mobile user $n$, which can be computed as $ t_{n}^{l} = \frac{\beta_{n}}{f_{n}^{l}}. $ To compute the energy consumption $E_{n}^{l}$ (in Joules) of the mobile user when the task is executed locally, we adopt the model in \cite{Mao2017_aSurveyMEC, Lyu2017Multiuser, Wang2017_JointComputation}. Specifically, $E_{n}^{l}$ can be derived as $ E_{n}^{l} = \kappa_{n} \beta_{n} \left( f_{n}^{l} \right)^{2}, $ where $\kappa_{n}$ is a coefficient relating to the chip's hardware architecture. According to the measurements in \cite{Wang2017_JointComputation}, we set $k_{n} = 5 \times 10^{-27}$. It is worth noting that $t_{n}^{l}$ and $E_{n}^{l}$ depend on unique features of the mobile user $n$ and the running application; therefore, they can be computed in advance. 

The computation overhead by the local computing approach  in terms of the computational time and energy consumption is computed as $ Z_{n}^{l} = \lambda_{n}^{t}t_{n}^{l} + \lambda_{n}^{e}E_{n}^{l}, $ where $\lambda_{n}^{t} \in \left[ 0,1 \right]$ and $\lambda_{n}^{e} + \lambda_{n}^{t} = 1$ are respectively weighted parameters of the computational time and energy consumption of the mobile user $n$. Similar to the heterogeneous computation tasks of mobile users, different mobile users may have different values of $\lambda_{n}^{t}$ and $\lambda_{n}^{e}$. The weighted parameters can affect the offloading decisions of mobile users. Consider a network scenario with three mobile users as an example, where the first mobile user with a latency-sensitive application sets $\lambda_{n}^{t} = 1$ and $\lambda_{n}^{e} = 0$, the second mobile user with an energy-hungry application and low battery state can set the weighted parameters $\lambda_{n}^{t} = 0$ and $\lambda_{n}^{e} = 1$, and the third mobile user can set $0 < \lambda_{n}^{t}, \lambda_{n}^{e} < 1$ if it takes both computational time and energy consumption into consideration of the offloading decision.

\subsection{Computation Model of MEC Servers}
In the case where a mobile user cannot execute the computation task due to a limited battery or application requirements, the mobile user will offload the computation task to the designated MEC server. To offload the computation task, a mobile user incurs extra overhead in terms of the time and energy consumption. The extra overhead in time is composed of the transmission time of the computation input data to the MEC server, the execution time of the computation task at the MEC server, and the transmission time of the computational result back to the mobile user. The extra overhead in energy consumption includes the energy consumption for computation offloading, execution of the computation task, and transmission of the computational result back to the mobile user. Since the focus of our work is on the perspective of mobile users and since MEC servers are generally powered by cable power supply \cite{CWang2017Joint, Chen2016Efficient}, we ignore the energy consumption for remote execution of the computation task. Moreover, the computational result is relatively small compared to the input data, so the time and energy consumption for transmission of the computational result back to the mobile user are therefore neglected. 

The time and energy costs for offloading the computation task $I_{n}$ are, respectively, computed as 
\begin{equation} \label{Eq:overhead_r_offloadingtime}
t_{n}^{\text{off}}\left( \boldsymbol{A}, \boldsymbol{P} \right) = \sum\limits_{m \in \mathcal{M}}\frac{\alpha_{n}a_{nm}}{R_{nm}}, \forall n \in \mathcal{N},
\end{equation}
where $a_{nm} = \sum\nolimits_{s \in \mathcal{S}}a_{nm}^{s}$, and 
\begin{equation} \label{Eq:overhead_r_offloadingenergy}
E_{n}^{\text{off}}\left( \boldsymbol{A}, \boldsymbol{P} \right) = p_{n}t_{n}^{\text{off}} = \frac{p_{n}}{\zeta_{n}} \alpha_{n}\sum\limits_{m \in \mathcal{M}}\frac{a_{nm}}{R_{nm}}, \forall n \in \mathcal{N},
\end{equation}
where $p_{n} = \boldsymbol{p}_{n}^{T}\boldsymbol{1}$ with $ \boldsymbol{X}^{T} $ being the normal transpose $ \boldsymbol{X} $, $\zeta_{n}$ is the power amplifier efficiency of the mobile user $n$. The execution time of the computation task $I_{n}$ is given by $ t_{n}^{\text{exe}}\left( \boldsymbol{A}, \boldsymbol{F}_{n} \right) = \sum\nolimits_{m \in \mathcal{M}}\frac{a_{nm}\beta_{n}}{f_{nm}}, \forall n \in \mathcal{N}, $ where $f_{nm}$ is the computational capability (in CPU cycles per second) that is assigned to mobile user $n$ by the MEC server $m$ in order to accomplish the task $I_{n}$. Here, $\boldsymbol{F} = \left\lbrace \boldsymbol{F}_{n} | n \in \mathcal{N} \right\rbrace$ is the computation resource profile, where $\boldsymbol{F}_{n} = \left\lbrace f_{nm} | \forall m \in \mathcal{M} \right\rbrace$ is the computation resource vector of the mobile user $n$. 

Since we assume that an MEC server is collocated with an SeNB in a small cell, the computational capability of an MEC server is often limited. Therefore, for each MEC server, the total computation resources assigned to all offloading users cannot exceed its maximum computational capability $f_{m}^{\max}$, i.e., the constraint, $ \sum\nolimits_{n \in \mathcal{N}}f_{nm} \leq f_{m}^{\max}, \forall m \in \mathcal{M} $, must be satisfied.

Similar to the computation overhead due to the local computing approach, overhead of the remote computing approach can be computed as $ Z_{n}^{r}\left( \boldsymbol{A}, \boldsymbol{P}, \boldsymbol{F}_{n} \right) = \lambda_{n}^{t}\left( t_{n}^{\text{off}} + t_{n}^{\text{exe}} \right) + \lambda_{n}^{e}E_{n}^{\text{off}}. $

\subsection{Problem Formulation}
Since our focus is to minimize the system-wide computation overhead in terms of the computational time and energy consumption, the objective function is defined as $Z\left( \boldsymbol{A}, \boldsymbol{P}, \boldsymbol{F} \right) = \sum\nolimits_{n \in \mathcal{N}} Z_{n}\left( \boldsymbol{A}, \boldsymbol{P}, \boldsymbol{F}_{n} \right)$, where $Z_{n}$ is given by $ Z_{n} = \left( 1 - \sum\nolimits_{m \in \mathcal{M}}a_{nm} \right) Z_{n}^{l} + \left( \sum\nolimits_{m \in \mathcal{M}}a_{nm} \right)Z_{n}^{r}, \forall n \in \mathcal{N}. $

For a given offloading decision profile $\boldsymbol{A}$, power allocation $\boldsymbol{P}$, and computation resource $\boldsymbol{F}$, as well as the objective function $Z\left(\cdot\right)$, we have the optimization formulation problem of joint computation offloading decision and resource allocation (OPT-JCORA) as follows:
\begin{align} 
& \min_{\left\lbrace \boldsymbol{A}, \boldsymbol{P}, \boldsymbol{F} \right\rbrace} Z\left( \boldsymbol{A}, \boldsymbol{P}, \boldsymbol{F} \right) \label{Eq:OPT_ObjectiveFunction} \\
& \text{s.t.} 
\: \text{C1:} \; 0 < p_{n}^{s} \leq p_{n}^{\max}, \forall s \in \mathcal{S}, \forall n \in \mathcal{N}_{\text{off}} \label{Eq:OPT_Constr_PowerBudget} \\
& 
\quad\; \text{C2:} \; a_{nm}^{s} = \left\lbrace 0, 1 \right\rbrace, \forall n \in \mathcal{N}, m \in \mathcal{M}, s \in \mathcal{S} \label{Eq:OPT_Constr_OffloadingDecision} \\ 
& 
\quad\; \text{C3:} \; \sum\limits_{m \in \mathcal{M}}\sum\limits_{s \in \mathcal{S}}a_{nm}^{s} \leq 1, \forall n \in \mathcal{N} \label{Eq:OPT_Constr_UE_OffloadingDecision} \\
& 
\quad\; \text{C4:} \; \sum\limits_{n \in \mathcal{N}}a_{nm}^{s} \leq 1, \forall m \in \mathcal{M}, s \in \mathcal{S} \label{Eq:OPT_Constr_MEC_OffloadingDecision} \\ 
& 
\quad\; \text{C5:} \; \sum\limits_{n \in \mathcal{N}}\sum\limits_{s \in \mathcal{S}}a_{nm}^{s} \leq q_{m}, \forall m \in \mathcal{M} \label{Eq:OPT_Constr_Quota} \\
& 
\quad\; \text{C6:} \; f_{nm} > 0, \forall n \in \mathcal{N}_{m}, \forall m \in \mathcal{M} \label{Eq:OPT_Constr_ComputingVariable} \\
& 
\quad\; \text{C7:} \; \sum\limits_{n \in \mathcal{N}}f_{nm} \leq f_{m}^{\max}, \forall m \in \mathcal{M}, \label{Eq:OPT_Constr_ComputingCapability}
\end{align}
where $p_{n}^{\max}$ is the maximum transmit power of the mobile user $n$, $\mathcal{N}_{m} = \left\lbrace n \in \mathcal{N} | a_{nm} = 1 \right\rbrace$ is the set of mobile users that offload their computation tasks to the MEC server $m$ and $\mathcal{N}_{\text{off}} = \bigcup\nolimits_{m \in \mathcal{M}}\mathcal{N}_{m}$ is the set of offloading mobile users that are not able to compute their tasks locally. In the optimization formulation~\eqref{Eq:OPT_Constr_UE_OffloadingDecision}, if $n \notin \mathcal{N}_{\text{off}}$, $p_{n} = 0$, i.e., the mobile user $ n $ executes the task locally. In addition, if $n \notin \mathcal{N}_{m}$, $f_{nm} = 0$, i.e., the MEC server $m$ does not assign any computation resource to the mobile user $n$. The constraint C1 makes sure that the transmit power of mobile user $n$ does not exceed the maximum value. The constraint C2 denotes that $\boldsymbol{A}$ is a binary vector. The constraints C3, C4, and C5 guarantee that a computation task is executed either locally or remotely, a subchannel is assigned to at most one mobile user, and the maximum number of mobile users $q_{m}$ can offload to an MEC server $m$, as explained in~\eqref{Eq:Constr_UE_OffloadingDecision}, \eqref{Eq:Constr_MEC_OffloadingDecision}, and~\eqref{Eq:Constr_Quota}, respectively. The constraint C6 denotes that the assigned computation resource from an MEC server to a mobile user is positive, and the constraint C7 makes sure that the total computation resource of an MEC server assigned to offloading users does not exceed its maximum value $f_{m}^{\max}$. 

The considered problem~\eqref{Eq:OPT_Constr_UE_OffloadingDecision} is difficult to solve due to the following reasons:
\begin{itemize}
\item There exist relationships among $\boldsymbol{A}$, $\boldsymbol{P}$, and $\boldsymbol{F}$. In addition, the data rate of mobile users in the denominator in~\eqref{Eq:overhead_r_offloadingtime} and~\eqref{Eq:overhead_r_offloadingenergy} makes the objective function highly complicated. Therefore, the objective function $U\left( \boldsymbol{A}, \boldsymbol{P}, \boldsymbol{F} \right)$ is not a convex function.
\item There are three set of optimization variables: offloading decision $\boldsymbol{A}$, power allocation $\boldsymbol{P}$, and computation resource $\boldsymbol{F}$. $\boldsymbol{P}$ and $\boldsymbol{F}$ are continuous variables while $\boldsymbol{A}$ is a binary variable; therefore, the feasible solution set of the problem~\eqref{Eq:OPT_Constr_UE_OffloadingDecision} is not convex.
\end{itemize}
To enable distributed computation offloading, in the next section, we will decompose problem~\eqref{Eq:OPT_Constr_UE_OffloadingDecision} into two parts: the computation offloading decision problem including the two phases of user association and subchannel assignment, and the resource allocation problem including the transmit power of mobile users and the computation resource allocation at MEC servers. 

\section{Proposed Algorithm}
\label{Sec:Algorithm}
\subsection{Problem Decomposition} \label{SubSec:ProblemDecomposition}
The OPT-JCORA is a mixed-integer and non-linear optimization problem since the offloading decision $\boldsymbol{A}$ is an binary vector and $\boldsymbol{P}$ and $\boldsymbol{F}$ are continuous vectors. In addition, the OPT-JCORA problem is NP-Hard \cite{Lyu2017Multiuser, Pochet2006Production}. As a result, it is difficult to obtain an optimal solution to the underlying problem~\eqref{Eq:OPT_Constr_UE_OffloadingDecision}. Observe from the OPT-JCORA problem~\eqref{Eq:OPT_Constr_UE_OffloadingDecision} that the resource constraints C1, C6, and C7 are decoupled from the computation offloading constraints C2, C3, C4, and C5; therefore, it is possible to decompose the OPT-JCORA problem into two subproblems: one for joint computation and communication resource allocation (JCCRA) and one for the computation offloading (CO) decision. The JCCRA subproblem is written as follows: 
\begin{align} 
& \min_{ \boldsymbol{P}, \boldsymbol{F} } \sum\limits_{n \in \mathcal{N}_{\text{off}}} Z_{n}\left( \boldsymbol{A}, \boldsymbol{P}, \boldsymbol{F}_{n} \right) \label{Eq:JCCRA_ObjectiveFunction} \\
& \text{s.t.} 
\: \text{C1, C6, C7}.
\end{align}
The objective value to the problem~\eqref{Eq:JCCRA_ObjectiveFunction}, defined as $Z(\boldsymbol{A})$, is a function of the offloading decision vector $\boldsymbol{A}$. Then, the CO subproblem is formulated as
\begin{align} 
& \min_{ \boldsymbol{A} } Z\left( \boldsymbol{A} \right) \label{Eq:CO_ObjectiveFunction} \\
& \text{s.t.} 
\: \text{C2, C3, C4, C5}.
\end{align}
By solving the two subproblems~\eqref{Eq:JCCRA_ObjectiveFunction} and~\eqref{Eq:CO_ObjectiveFunction} sequentially in each iteration until convergence, the solution to the underlying problem~\eqref{Eq:OPT_Constr_UE_OffloadingDecision} can be finally obtained. The proposed framework for solving the OPT-JCORA problem is summarized in Fig.~\ref{Fig:Proposed_Framework}.

\subsection{Computation Offloading As a Matching Problem} \label{SubSec:ComputationOffloading}
We consider the offloading decision problem for a given $\boldsymbol{P}$ and $\boldsymbol{F}$ by solving the following optimization problem:
\begin{align} 
& \min_{\boldsymbol{A}} Z\left( \boldsymbol{A} \right) \label{Eq:OPT_CO_ObjectiveFunction} \\
& \text{s.t.} 
\: a_{nm}^{s} = \left\lbrace 0, 1 \right\rbrace, \forall n \in \mathcal{N}, m \in \mathcal{M}, s \in \mathcal{S} \label{Eq:OPT_CO_Constr_OffloadingDecision} \\ 
& 
\quad\; \sum\limits_{m \in \mathcal{M}}\sum\limits_{s \in \mathcal{S}}a_{nm}^{s} \leq 1, \forall n \in \mathcal{N} \label{Eq:OPT_CO_Constr_UE_OffloadingDecision} \\
& 
\quad\; \sum\limits_{n \in \mathcal{N}}a_{nm}^{s} \leq 1, \forall m \in \mathcal{M}, s \in \mathcal{S} \label{Eq:OPT_CO_Constr_MEC_OffloadingDecision}  \\
&
\quad\; \sum\limits_{n \in \mathcal{N}}\sum\limits_{s \in \mathcal{S}}a_{nm}^{s} \leq q_{m}, \forall m \in \mathcal{M} \label{Eq:OPT_CO_Constr_Quota} \\ 
& 
\quad\;  p_{n}^{s} = p_{n}^{\max}/S, \forall n \in \mathcal{N} \label{Eq:OPT_CO_Constr_UniformPower}  \\
& 
\quad\;  f_{nm} = f_{m}^{\max}/q_{m}, \forall n \in \mathcal{N}, m \in \mathcal{M}.  \label{Eq:OPT_CO_Constr_UniformComputationResource}
\end{align}
In order to evaluate the average contribution of each mobile user to the objective $Z(\boldsymbol{A})$, it is assumed that the total transmit power of each mobile user is divided equally among subchannels and the total computation resources of each MEC server is uniform among the maximum $q_{m}$ offloading users, i.e., $p_{n}^{s} = p_{n}^{\max}/S$ and $f_{nm} = f_{m}^{\max}/q_{m}$, as illustrated in constraints~\eqref{Eq:OPT_CO_Constr_UniformPower} and~\eqref{Eq:OPT_CO_Constr_UniformComputationResource}, respectively.

\begin{figure}[!ht]
	\centering
	\includegraphics[width=0.95\linewidth]{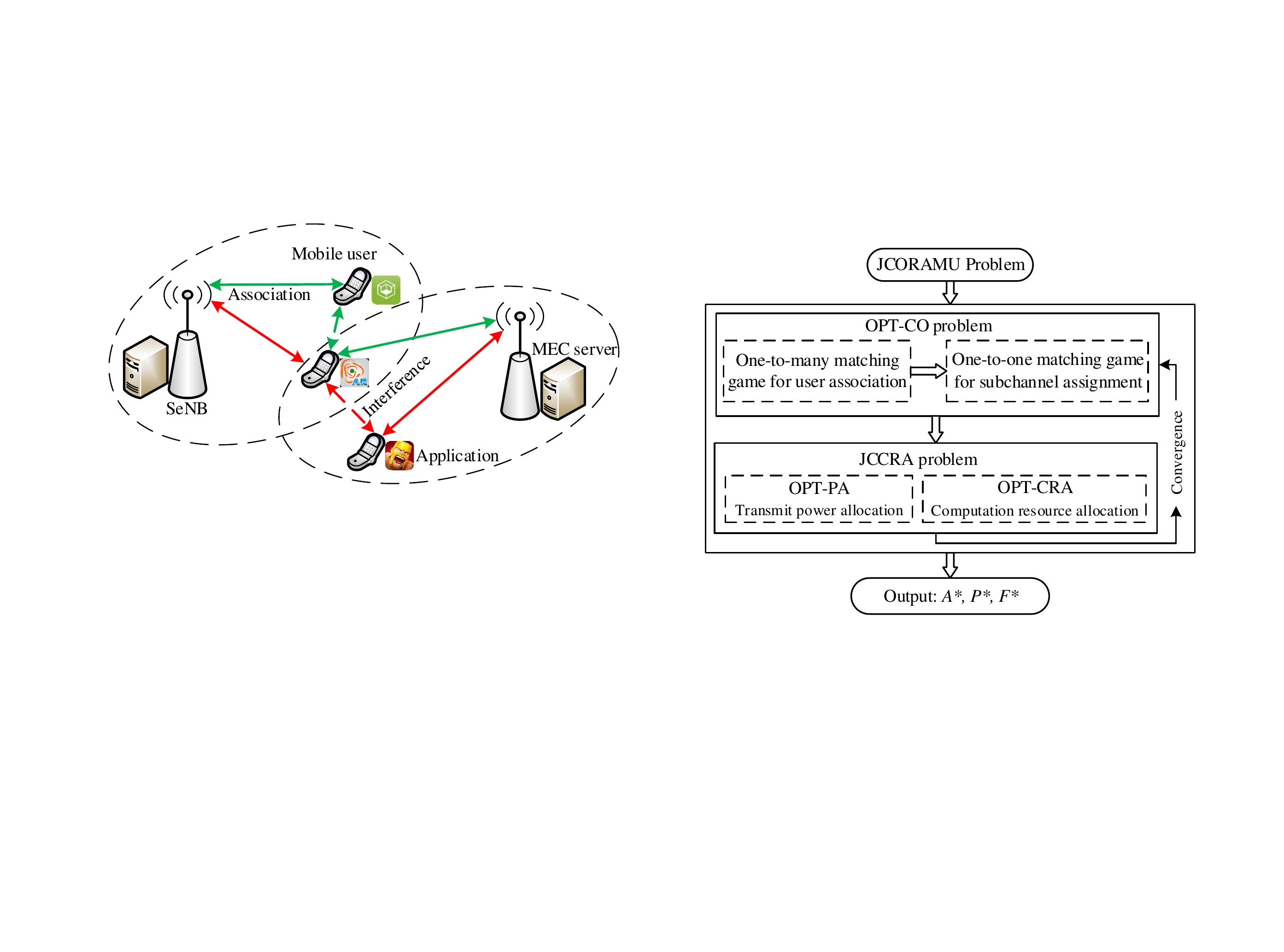}
	\caption{Proposed framework.}
	\label{Fig:Proposed_Framework}
\end{figure}

The OPT-CO subproblem is a mixed integer nonlinear programming (MINP) as well as an NP-Hard problem \cite{Lyu2017Multiuser, Pochet2006Production}. To solve the OPT-CO problem, the objective function~\eqref{Eq:OPT_CO_ObjectiveFunction} is boiled down to the following:
\begin{align} \label{Eq:RewrittenObjective}
& Z\left( \boldsymbol{A} \right) = \sum\limits_{n \in \mathcal{N}}Z_{n}^{l} + \\
& \sum\limits_{n \in \mathcal{N}}a_{n}\left[ \sum\limits_{m \in \mathcal{M}}a_{nm} \left( \frac{\lambda_{n}^{t} \alpha_{n} + \lambda_{n}^{e}p_{n}\alpha_{n}\zeta_{n}^{-1}}{R_{nm}} + \frac{\lambda_{n}^{t}\beta_{n}}{f_{nm}} \right) - Z_{n}^{l} \right],
\end{align}
where $a_{n} = \sum\nolimits_{m \in \mathcal{M}}a_{nm}$, whose value denotes the offloading decision of mobile user $n$, i.e., $a_{n} = 1$ ($a_{n} = 0$) indicates that the mobile user $n$ offloads (executes locally). Recall that this paper considers binary offloading such that a task cannot be partitioned into subtasks. To be executed remotely at the MEC server, a computation offloading has to be profitable to the mobile user in terms of the execution latency and/or energy consumption. It is therefore observed from~\eqref{Eq:RewrittenObjective} that the mobile user $n$ executes the task locally if the following condition holds
\begin{equation} \label{Eq:Condition_Offloading}
\Upsilon_{n} \vcentcolon= \left( \frac{\lambda_{n}^{t} \alpha_{n} + \lambda_{n}^{e}p_{n}^{\min}\alpha_{n}\zeta_{n}^{-1}}{R_{nm}^{\max}} + \frac{\lambda_{n}^{t}\beta_{n}}{f_{0}} \right) - Z_{n}^{l} \geq 0,
\end{equation}
where $R_{nm}^{\max} = B_{s}\log_{2}\left( 1+ p_{n}\max\nolimits_{m \in \mathcal{M},s \in \mathcal{S}}\left\lbrace h_{nm}^{s} \right\rbrace / n_{0} \right)$ and $f_{0} = \max\nolimits_{m \in \mathcal{M}}\left\lbrace f_{m}^{\max} \right\rbrace$. Let $\mathcal{N}_{\text{loc}} = \left\lbrace n \in \mathcal{N} | \Upsilon_{n} \geq 0 \right\rbrace$ be the set of mobile users that execute their tasks locally, and $\mathcal{N}_{\text{pof}} = \left\lbrace n \in \mathcal{N} | \Upsilon_{n} < 0 \right\rbrace$ be the set of mobile users that potentially offload their tasks to the MEC servers. To find an offloading solution for the remaining mobile users, we temporarily assume that $a_{n} = 1, \forall n \in N_{\text{pof}}$, ignore the fixed parts, and rewrite the objective function in~\eqref{Eq:RewrittenObjective}, as follows: 
\begin{equation} \label{Eq:RewrittenObjective_Matching}
Z\left( \boldsymbol{A} \right) = \sum\limits_{n = 1}^{\left| \mathcal{N}_{\text{pof}} \right|}\sum\limits_{m = 1}^{M}a_{nm} \left( \frac{\lambda_{n}^{t} \alpha_{n} + \lambda_{n}^{e}p_{n}\alpha_{n}\zeta_{n}^{-1}}{ \sum\limits_{s = 1}^{S}a_{nm}^{s}B_{s}\log_{2}\left( 1 + \Gamma_{nm}^{s} \right) } + \frac{\lambda_{n}^{t}\beta_{n}}{f_{nm}} \right).
\end{equation}
From the new objective~\eqref{Eq:RewrittenObjective_Matching}, the computation offloading decision problem can now be further decomposed into two subproblems: 1) \textit{which MEC server does a  mobile user offload to} and 2) \textit{which subchannel does a mobile user utilize to offload the task}. Let $y_{nm}^{s} = a_{nm}^{s}, \forall n \in \mathcal{N}_{\text{pof}}, m \in \mathcal{M}, s \in \mathcal{S}$ and $\boldsymbol{Y} = \left[ y_{nm}^{s} \right]_{\left| \mathcal{N}_{\text{pof}} \right| \times M \times S}$. For a given $\boldsymbol{Y}$, the first subproblem finds the matching $\boldsymbol{X} = \left[ a_{nm} \right]_{\left| \mathcal{N}_{\text{pof}} \right| \times M}$ between $\left| \mathcal{N}_{\text{pof}} \right|$ mobile users and $M$ MEC servers, in order to minimize the objective function $Z(\boldsymbol{A})$. For a given $\boldsymbol{X}$, the second subproblem is to determine the matching between mobile users $\mathcal{N}_{m}$ that offload their tasks to the same MEC server $m$ and subchannels $\mathcal{S}$, with the objective of maximizing the achievable offloading rate $R_{nm}^{y}$.

In what follows, we propose distributed processes based on a one-to-many matching game to find which MEC server a mobile user offloads to, i.e, the user association, and based on a one-to-one matching game to find which subchannel a mobile user utilizes to offload the task, i.e., the subchannel assignment, \cite{Bayat2016Matching, Gu2015Matching}. In the first matching game, there are two types of players: mobile users and MEC servers. The strategy of mobile users is to select the best MEC server to maximize the benefit of computation offloading and requesting computation offloading, and the strategy of MEC servers is to either accept or reject the offloading requests from mobile users. From the constraints, each mobile user can offload the task to at most one MEC server according to~\eqref{Eq:OPT_CO_Constr_UE_OffloadingDecision} and each MEC server can execute multiple tasks from mobile users according to~\eqref{Eq:OPT_CO_Constr_Quota}. The two types of players in the second matching game are mobile users and subchannels. The strategy of mobile users is to select the most preferred subchannel in order to maximize the offloading rate, and that of the subchannels is to make a decision on either accepting or rejecting the bids from mobile users. In this second game, each mobile user is assumed to offload on at most one subchannel and each subchannel can be matched with at most one offloading user. 

\subsubsection{Matching Game for User Association}
The one-to-many matching game is defined formally below. 
\begin{definition}
Given two disjoint sets of players, $\mathcal{M}$ and $\mathcal{N}_{\text{pof}}$, a one-to-many matching function $\Psi:\mathcal{N}_{\text{pof}} \rightarrow \mathcal{M}$ is defined such that for all $n \in \mathcal{N}_{\text{pof}}$ and $m \in \mathcal{M}$
\begin{enumerate}
\item $\Psi(m) \subseteq \mathcal{N}_{\text{pof}}$ and $|\Psi(m)| \leq q_{m}$;
\item $\Psi(n) \subseteq \mathcal{M}$ and $|\Psi(n)| \in \lbrace 0, 1 \rbrace$;
\item $m = \Psi(n) \leftrightarrow n = \Psi(m)$.
\end{enumerate}
\end{definition}
This matching game is defined by a tuple $\left( \mathcal{M}, \mathcal{N}_{\text{pof}}, \mathcal{q} \right)$, with $\mathcal{q} = \lbrace q_{m} \vert \forall m \in \mathcal{M} \rbrace$ being the MEC servers' quota vector. The first condition implies that each MEC server $m$ can execute at most $q_{m}$ computation tasks as in~\eqref{Eq:OPT_CO_Constr_Quota}, the second condition indicates that each mobile user can offload the task to at most one MEC server, and the third condition implies that if the mobile user $n$ is matched with the MEC server $m$, then the MEC server $m$ is also matched with the mobile user $n$. The output of this game is a user association mapping $\Psi$ between mobile users and MEC servers. 

Next, we define $\phi_{n,\text{UA}}(m)$ as the preference of mobile user $n$ for MEC server $m$ and $\phi_{m,\text{UA}}(n)$ as the preference of MEC server $m$ for mobile user $n$. We also define $\succ_{n,\text{UA}}$ and $\succ_{m,\text{UA}}$ as the preference relations of mobile user $n$ and MEC server $m$, respectively. The notation $m_{1} \succ_{n,\text{UA}} m_{2}$ implies that the mobile user $n$ prefers the MEC server $m_{1}$ to $m_{2}$, i.e., $\phi_{n,\text{UA}}(m_{1}) > \phi_{n,\text{UA}}(m_{2})$. Similarly, the notation $n_{1} \succ_{m,\text{UA}} n_{2}$ implies that the MEC server $m$ prefers the mobile user $n_{1}$ to $n_{2}$ if the computation overhead with $n_{1}$ is smaller than that with $n_{2}$, i.e., $\phi_{m,\text{UA}}(n_{1}) < \phi_{m,\text{UA}}(n_{2})$. For the association matching game, the preference lists of mobile users and MEC servers are defined, as follows.

\noindent \textbf{Preference of the mobile user}: For user association problems in wireless multicell networks, the average SINR over all subchannels is considered as one of the most common criteria \cite{Bayat2014Distributed, Liu2016User, LeAnh2017Matching}. In this paper, the preference value of the mobile user $n$ when it offloads the task $I_{n}$ to the MEC server $m$ is defined as 
\begin{equation} \label{Eq:Preference_UA_MobileUser}
\phi_{n,\text{UA}}(m) = \varphi_{\text{UA}}\alpha_{n}^{-1}\log_{2} \left( 1 + \sum\limits_{s \in \mathcal{S}}\Gamma_{nm}^{s} \right) + \varepsilon_{\text{UA}}\beta_{n}^{-1}f_{nm},
\end{equation}
where $\varphi_{\text{UA}}$ and $\varepsilon_{\text{UA}}$ are two weighted parameters, $\Gamma_{nm}^{s}$ is specified in~\eqref{Eq:SINR} with the set $\mathcal{N}$ replaced by $\mathcal{N}_{\text{pof}}$, and $f_{nm} = f_{m}^{\max}/q_{m}$. Intuitively, the user $n$ prefers the MEC server $m_{1}$ to $m_{2}$ if the mobile user $n$ has a higher offloading rate and computation resource with $m_{1}$ than with $m_{2}$. 

\noindent \textbf{Preference of the MEC server}: We define the preference value of the MEC server $m$ when it executes the computation task from mobile user $n$ as the computation overhead, which is expressed as  
\begin{equation} \label{Eq:Preference_UA_MECServer}
\phi_{m,\text{UA}}(n) = \frac{\lambda_{n}^{t} \alpha_{n} + \lambda_{n}^{e}p_{n}\alpha_{n}\zeta_{n}^{-1}}{R_{nm}} + \frac{\lambda_{n}^{t}\beta_{n}}{f_{nm}},
\end{equation}
where $R_{nm} = B_{s}\log_{2} \left( 1 + \sum\nolimits_{s \in \mathcal{S}}\Gamma_{nm}^{s} \right)$. We say that the MEC server $m$ prefers the mobile user $n_{1}$ to $n_{2}$ (when $n_{1}$ and $n_{1}$ select the same MEC server $m$) if the mobile user $n_{2}$ has lower computation overhead than the mobile user $n_{2}$. 

\subsubsection{Distributed Algorithm for User Association}
Now, we propose a distributed algorithm to find the matching between mobile users and MEC servers while minimizing the computation overhead. The specific details of the proposed algorithm are given in Alg.~\ref{Alg:User_MECServer}, with new notations defined and described below.
\begin{algorithm}
\caption{Users - MEC servers matching algorithm}\label{Alg:User_MECServer}
\begin{algorithmic}[1]
\State \textbf{Initialization} 
\begin{enumerate}
\item $\mathcal{M}$, $\mathcal{N}_{\text{pof}}$, $b_{nm}^{\text{UA}} = 0, \forall n \in \mathcal{N}_{\text{pof}}$.
\item Construct the preference for all mobile users in  $\mathcal{N}_{\text{pof}}$ via~\eqref{Eq:Preference_UA_MobileUser} and the preference list $\mathcal{M}_{n} = \mathcal{M}, \forall n \in \mathcal{N}_{\text{pof}}$, set $\mathcal{N}_{\text{unmatched}} = \mathcal{N}_{\text{pof}}$, and initialize the list of requested users $\mathcal{N}_{m}^{\text{req}} = \emptyset$ and the list of rejected users $\mathcal{N}_{m}^{\text{rej}} = \emptyset, \forall m \in \mathcal{M}$.
\end{enumerate}
\State \textbf{Find a stable matching $\Psi^{*}$}
\While {$\sum\nolimits_{m \in \mathcal{M}}\sum\nolimits_{n \in \mathcal{N}_{\text{pof}}}b_{nm}^{\text{UA}} \neq 0$}
\For {$n = 1$ to $|\mathcal{N}_{\text{unmatched}}|$}
\State Find $m = \argmax\limits_{m \in \mathcal{M}_{n}}\phi_{n,\text{UA}}(m)$.
\State Send a request to the server $m$ by setting $b_{nm}^{\text{UA}} = 1$.
\EndFor
\For {$m = 1$ to $M$}
\State Update $\mathcal{N}_{m}^{\text{req}} \leftarrow \lbrace n: b_{nm}^{\text{UA}} = 1, \forall n \in \mathcal{N}_{\text{pof}}\rbrace$.
\State Construct the preference according to~\eqref{Eq:Preference_UA_MECServer}.
\If {$|\mathcal{N}_{m}^{\text{req}}| \leq q_{m}$}
\State Update $\mathcal{N}_{m} \leftarrow \mathcal{N}_{m}^{\text{req}}$.
\Else
\Repeat
\State Accept $n = \argmax\limits_{n \in \succ_{m,\text{UA}}} \sum\limits_{n \in \mathcal{N}_{m}^{\text{req}}} \phi_{m,\text{UA}}(n)$.
\State Update $\mathcal{N}_{m} \leftarrow \mathcal{N}_{m} \cup n$.
\Until{$|\mathcal{N}_{m}| = q_{m}$}
\EndIf
\State Update $\mathcal{N}_{m}^{\text{rej}} \leftarrow \lbrace \mathcal{N}_{m}^{\text{req}} \setminus N_{m} \rbrace$.
\State Update $\mathcal{M}_{n} \leftarrow \lbrace \mathcal{M}_{n} \setminus m \rbrace, \forall n \in \mathcal{N}_{m}^{\text{rej}}$.
\EndFor
\State Update $\mathcal{N}_{\text{unmatched}} \leftarrow \mathcal{N}_{\text{unmatched}} \cap \lbrace \mathcal{N}_{1}^{\text{rej}} \cup ... \cup \mathcal{N}_{M}^{\text{rej}}\rbrace \rbrace$.
\EndWhile
\State \textbf{End of the algorithm}: output is a stable matching $\Psi^{*}$.
\end{algorithmic}
\end{algorithm}

First, the list of potential MEC servers for each mobile user $n$ is defined as $\mathcal{M}_{n}$ and initialized as $\mathcal{M}$, and the sets of unmatched users is $\mathcal{N}_{\text{unmatched}} = \mathcal{N}_{\text{pof}}$, the set of rejected and requested users for each MEC server $m$ are defined as empty sets $\mathcal{N}_{m}^{\text{rej}}$ and $\mathcal{N}_{m}^{\text{req}}$, respectively. In addition, each mobile user constructs its preference over all potential MEC servers according to Eq.~\eqref{Eq:Preference_UA_MobileUser}. 

Next, each mobile user $n$ decides the best MEC server, which has the largest preference among $\mathcal{M}_{n}$, by using its preference relation in line 5 and sends a bit request $b_{nm}^{\text{UA}}$ to the MEC server $m$. The bid function $b_{nm}^{\text{UA}}$ is $1$ if the user $m$ offloads to the MEC server $m$ and $0$ otherwise (line 6). 

After bidding of all mobile users, each MEC server collects the bid requests from mobile users and updates the list of requested users (line 9) and constructs the preference over all requested users according to Eq.~\eqref{Eq:Preference_UA_MECServer}. The MEC server $m$ is able to accept all of the requested users if the number of requested users is smaller than its quota $q_{m}$ (lines 12), otherwise it will select $q_{m}$ among $\left| \mathcal{N}_{m}^{\text{req}} \right|$ requested users (lines 15 and 16). Unmatched users are then inserted into the set of rejected users of each MEC server. Meanwhile, each MEC server is removed from the preference list of its rejected users. Finally, the list of unmatched users is updated (line 22). Once there is no further requested user (condition checking at line 3), the algorithm stops. 

If a stable matching does not exist, it is computationally difficult to find the solution. Fortunately, the outcome of Alg.~\ref{Alg:User_MECServer} is a stable matching $\Psi^{*}$. To explain how the matching Alg.~\ref{Alg:User_MECServer} achieves a stable matching, we present the definitions of a blocking pair and a stable matching in Definitions~\ref{Def:BlockingPair} and~\ref{Def:StableMatching}, respectively \cite{Bayat2016Matching, Li2017Matching}.
\begin{definition}[Blocking Pair] \label{Def:BlockingPair}
The pair $\left( m_{0},n_{0} \right)$ is a blocking pair for the matching $\Psi$, only if $m_{0} \succ_{n,\text{UA}} m$, $m \in \Psi(n_{0})$ and $n_{0} \succ_{m,\text{UA}} n$, $n \in \Psi(m_{0})$, for $m_{0} \notin \Psi(n_{0})$ and $n_{0} \notin \Psi(m_{0})$. In other words, there exists a partnership $\left( m_{0},n_{0} \right)$ such that $ m_{0} $ and $ n_{0} $ are not matched with each other under the current matching $\Psi$ but prefer to be matched with each other.
\end{definition}
\begin{definition}[Stable Matching] \label{Def:StableMatching}
A matching $\Psi$ is said to be stable if it admits no blocking pair.
\end{definition}
\begin{theorem}
The matching $\Psi^{*}$ generated by Alg.~\ref{Alg:User_MECServer} is stable and guarantees a local optimal solution to the underlying problem.
\end{theorem}
\begin{proof}
For a given transmit power allocation, computation resource allocation, and subchannel assignment, the preference of each mobile user and MEC server is fixed. Therefore, Alg.~\ref{Alg:User_MECServer} is known as the deferred acceptance algorithm in the two-sided matching problem between mobile users and MEC servers, which guarantees a stable matching \cite{Gu2015Matching}. The first part is proved. 

At each iteration $ t $, the outcome of Alg.~\ref{Alg:User_MECServer} maps to a user-server association $ \boldsymbol{X}^{\left(t\right)} $, which captures to the objective 
\begin{equation}
Z\left( \boldsymbol{X}^{\left(t\right)} \right) = \sum\limits_{n = 1}^{\left| \mathcal{N}_{\text{pof}} \right|}\sum\limits_{m = 1}^{M}x_{nm}^{(t)} \left( \frac{\lambda_{n}^{t} \alpha_{n} + \lambda_{n}^{e}p_{n}\alpha_{n}\zeta_{n}^{-1}}{R_{nm}} + \frac{\lambda_{n}^{t}\beta_{n}}{f_{nm}} \right).
\end{equation}
Assume there exists a blocking pair $\left( m_{0}, n_{0} \right)$ at iteration $ t $ such that the preference of the MEC server $m$ and mobile user $n$ can be improved when $\left( m_{0}, n_{0} \right)$ is added to the current matching $\Psi$. Accordingly, $ Z( \boldsymbol{X}^{\left(t\right)} ) > Z( \boldsymbol{X}^{\left(t+1\right)} )$, i.e., the computation overhead is reduced. According to Definition~\ref{Def:StableMatching}, there is no blocking pair at the final matching of the algorithm. As a result, the matching algorithm converges to a local optimal solution to the underlying problem. The second part is proved.
\end{proof}

\subsubsection{Matching Game for Subchannel Assignment}
After determining the user association mapping $\Psi$ or $\mathcal{N}_{m}, m \in \mathcal{M}$, the one-to-one matching game is defined to find the subchannel assignment\footnote{Here, we omit the subscript of MEC server $m$ and consider the matching definition for the set of mobile users $\mathcal{N}_{m}$ associated with MEC server $m$ and the set of subchannel $\mathcal{S}$.}. From~\eqref{Eq:RewrittenObjective_Matching}, the problem for subchannel assignment of SeNB $m$ can be expressed as
\begin{align} 
& \max_{\boldsymbol{Y}} \sum\limits_{s = 1}^{S}\left[ R_{nm}^{s} = B_{s}\log_{2}\left( 1 + \Gamma_{nm}^{s} \right) \right] \label{Eq:OPT_CA_ObjectiveFunction} \\
& \text{s.t.} 
\: y_{nm}^{s} = \left\lbrace 0, 1 \right\rbrace, \forall n \in \mathcal{N}_{m}, s \in \mathcal{S} \label{Eq:OPT_CA_Constr_OffloadingDecision} \\ 
& 
\quad\; \sum\limits_{s \in \mathcal{S}}y_{nm}^{s} \leq 1, \forall n \in \mathcal{N}_{m} \label{Eq:OPT_CA_Constr_UE_OffloadingDecision} \\
& 
\quad\; \sum\limits_{n \in \mathcal{N}_{m}}y_{nm}^{s} \leq 1, s \in \mathcal{S} \label{Eq:OPT_CA_Constr_MEC_OffloadingDecision}  \\
& 
\quad\;  p_{n}^{s} = p_{n}^{\max}/S, \forall n \in \mathcal{N}_{m}. \label{Eq:OPT_CA_Constr_UniformPower} 
\end{align} 
\begin{definition}
Given two disjoint sets of players, $\mathcal{N}_{m}$ and $\mathcal{S}$, a one-to-one matching function $\Omega:\mathcal{N}_{m} \rightarrow \mathcal{S}$ is defined such that for all $n \in \mathcal{N}_{m}$ and $s \in \mathcal{S}$
\begin{enumerate}
\item $\Omega(s) \subseteq \mathcal{N}_{m}$ and $|\Omega(s)| \in \lbrace 0, 1 \rbrace$;
\item $\Omega(n) \subseteq \mathcal{S}$ and $|\Omega(n)| \in \lbrace 0, 1 \rbrace$;
\item $n = \Omega(s) \leftrightarrow s = \Omega(n)$.
\end{enumerate}
\end{definition}
The first two conditions ensure that each mobile user can utilize at most one subchannel and a subchannel is assigned to at most one mobile user, as illustrated in~\eqref{Eq:OPT_CA_Constr_UE_OffloadingDecision} and~\eqref{Eq:OPT_CA_Constr_MEC_OffloadingDecision}, respectively, and condition 3 implies that if mobile user $n$ is matched with subchannel $s$, then subchannel $s$ is also matched with mobile user $n$. The outcome of this one-to-one matching game is the association mapping $\Omega$ between the set of mobile users $\mathcal{N}_{m}$ and the set of subchannels $\mathcal{S}$. Similar to the matching definition for user association, we define $\phi_{n,\text{CA}}(s)$ as the preference of mobile user $n$ for subchannel $s$ and $\phi_{s,\text{CA}}(n)$ as the preference of subchannel $s$ for mobile user $n$. Then, the notation $s_{1} \succ_{n,\text{CA}} s_{2}$ denotes that mobile user $n$ prefers subchannel $s_{1}$ to $s_{2}$, i.e., $\phi_{n,\text{CA}}(s_{1}) > \phi_{n,\text{CA}}(s_{1})$, and the notation $n_{1} \succ_{s,\text{CA}} n_{2}$ indicates that subchannel $s$ prefers mobile user $n_{1}$ to $n_{2}$, i.e., $\phi_{s,\text{CA}}(n_{1}) > \phi_{s,\text{CA}}(n_{2})$. For the subchannel assignment game, the preference lists of mobile users and subchannels are defined as follows: 

\noindent \textbf{Preference of the mobile user}: After determining the MEC server selection, the mobile user $m$ achieves the following preference when it accesses to the subchannel $s$
\begin{equation} \label{Eq:Preference_CA_User}
\phi_{n,\text{CA}}(s) = R_{nm}^{s}.
\end{equation}
The preference in~\eqref{Eq:Preference_CA_User} implies that 1) subchannel selection of a user only affects the achievable offloading rate, which in turn determines the offloading time, as illustrated in~\eqref{Eq:overhead_r_offloadingtime}, and 2) each mobile user prefers to offload over the subchannel that offers a higher offloading rate. 

\noindent \textbf{Preference of the MEC server for subchannels}:
The preference of the MEC server $m$ on subchannel $s$ to be matched with mobile user $n$ can be written as
\begin{equation} \label{Eq:Preference_CA_Subchannel}
\phi_{s,\text{CA}}(n) = \varphi_{\text{CA}}R_{nm}^{s} - \sum\nolimits_{m' \in \mathcal{M}, m' \neq m}\delta_{m'}^{s}g_{nm'}^{s}p_{n}^{s},
\end{equation}
where $\varphi_{\text{CA}}$ and $\delta_{m}^{s}$ are two weighted coefficients. The preference in~\eqref{Eq:Preference_CA_Subchannel} implies that the SeNB $m$ assigns the subchannel $s$ to the mobile user $n$ so as to maximize the achievable offloading rate of that user and minimize the aggregated interference to other SeNBs. 

\subsubsection{Distributed Algorithm for Subchannel Assignment}
Similar to the association matching, a distributed algorithm is designed to allocate subchannels of an SeNB to its associated users. The specific details of the algorithm are summarized in Alg.~\ref{Alg:User_Subchannels}. First, we obtain the sets of mobile users $\mathcal{N}_{m}$ from Alg.~\ref{Alg:User_MECServer} and the subchannels $\mathcal{S}$, and initialize the set of unmatched users $\mathcal{N}_{\text{unmatched}}$, the set of potential subchannels for each user $\mathcal{S}_{n}$, the lists of requested users $\mathcal{N}_{s}^{\text{req}}$ and rejected users $\mathcal{N}_{s}^{\text{rej}}$ for each subchannel $s$. Each user $n$ also constructs its preference over all potential subchannels $\mathcal{S}_{n}$ (step 3 in Initialization). Next, each mobile user $n$ selects the best subchannel $s$ (line 5) and sends an access request for subchannel $s$ to SeNB $m$ (line 6). Here, the bid value $b_{ns}^{\text{CA}}$ is set to $1$ if the mobile user $n$ bids for the subchannel $s$ and $0$ otherwise. At SeNB $m$, the list of requested users to each subchannel $s$ is updated (line 9). Then, the MEC server $m$ selects the best user among the $\left| \mathcal{N}_{s}^{\text{req}} \right|$ requested users for each subchannel $s$ (line 11) and assigns the subchannel $s$ to that user (line 12). After that, the list of unmatched users for each subchannel $s$, $\mathcal{N}_{s}^{\text{rej}}$, is updated (line 13) and each subchannel $s$ is removed from the list of potential subchannels of its rejected users $\mathcal{N}_{s}^{\text{rej}}$ (line 14). Based on the list of rejected users for all subchannels, the list of unmatched users is also updated (line 16). The algorithm stops if there is no further bidding between mobile users $\mathcal{N}_{m}$ and subchannels $\mathcal{S}$ (line 3). 
\begin{algorithm}
\caption{Users - subchannels matching algorithm}\label{Alg:User_Subchannels}
\begin{algorithmic}[1]
\State \textbf{Initialization} 
\begin{enumerate}
\item $\mathcal{N}_{m}$, $\mathcal{S}$, $b_{ns}^{\text{CA}} = 0, \forall n \in \mathcal{N}_{m}$.
\item Set $\mathcal{N}_{\text{unmatched}} = \mathcal{N}_{m}$, $\mathcal{S}_{n} = \mathcal{S}, \forall n \in \mathcal{N}_{m}$, the list of requested users $\mathcal{N}_{s}^{\text{req}} = \emptyset$ and the list of rejected users $\mathcal{N}_{s}^{\text{rej}} = \emptyset, \forall s \in \mathcal{S}$.
\item Construct the preference for all mobile users in  $\mathcal{N}_{m}$ via~\eqref{Eq:Preference_CA_User}.
\end{enumerate}
\State \textbf{Find a stable matching $\Omega^{*}$}
\While {$\sum\limits_{s \in \mathcal{S}}\sum\limits_{n \in \mathcal{N}_{m}}b_{ns}^{\text{UA}} \neq 0$}
\For {$n = 1$ to $|\mathcal{N}_{\text{unmatched}}|$}
\State Find $s = \argmax\limits_{s \in \mathcal{S}_{n}}\phi_{n,\text{CA}}(s)$.
\State Send a request to the server $m$ by setting $b_{ns}^{\text{CA}} = 1$.
\EndFor
\For {$s = 1$ to $S$}
\State Update $\mathcal{N}_{s}^{\text{req}} \leftarrow \lbrace n: b_{ns}^{\text{CA}} = 1, \forall n \in \mathcal{N}_{m}\rbrace$.
\State Construct the preference via~\eqref{Eq:Preference_CA_Subchannel}.
\State Find $n = \argmax\limits_{n \in \mathcal{N}_{s}^{\text{req}}}\phi_{s,\text{CA}}(n)$.
\State Assign the subchannel $s$ to the mobile user $n$.
\State Update $\mathcal{N}_{s}^{\text{rej}} \leftarrow \lbrace \mathcal{N}_{s}^{\text{req}} \setminus n \rbrace$.
\State Update $\mathcal{S}_{n} \leftarrow \lbrace \mathcal{S}_{n} \setminus s \rbrace, \forall n \in \mathcal{N}_{s}^{\text{rej}}$.
\EndFor
\State Update $\mathcal{N}_{\text{unmatched}} \leftarrow \mathcal{N}_{\text{unmatched}} \cap \lbrace \mathcal{N}_{1}^{\text{rej}} \cup ... \cup \mathcal{N}_{S}^{\text{rej}}\rbrace \rbrace$.
\EndWhile
\State \textbf{End of the algorithm}: outcome is a stable matching $\Omega^{*}$.
\end{algorithmic}
\end{algorithm}

\begin{theorem} \label{Theo:Convergence_StableMatching_CA}
The matching $\Omega^{*}$ generated by Alg.~\ref{Alg:User_Subchannels} is stable and can achieve a local maximum of the problem~\eqref{Eq:OPT_CA_Constr_UE_OffloadingDecision}.
\end{theorem}
\begin{proof}
The definitions of a blocking pair and stable matching for the matching problem between mobile users and subchannels are similar to those in Definition~\ref{Def:BlockingPair} and~\ref{Def:StableMatching}, respectively. For a given transmit power allocation, computation resource allocation, and association matching $\Psi^{*}$, the matching in Alg.~\ref{Alg:User_Subchannels} has the nature of deferred acceptance. Thus, a stable matching $\Omega^{*}$ can be found by Alg.~\ref{Alg:User_Subchannels}.
Note that the outcome of Alg.~\ref{Alg:User_Subchannels} at each iteration $ t $ maps to a subchannel assignment $ \boldsymbol{Y}^{\left(t\right)} $ and the objective for a $ \boldsymbol{Y}^{\left(t\right)} $ is $ R_{m}(\boldsymbol{Y}^{\left(t\right)}) = \sum\nolimits_{s = 1}^{S} R_{nm}^{s}(\boldsymbol{Y}^{\left(t\right)}) $. Moreover, the matching at iteration $ t+1 $ guarantees that $ R_{m}(\boldsymbol{Y}^{\left(t\right)}) \leq R_{m}(\boldsymbol{Y}^{\left(t + 1\right)}) $, i.e., the objective is monotonically improved during the matching process. Consequently, Theorem~\ref{Theo:Convergence_StableMatching_CA} is proved. 
\end{proof}

\subsection{Joint CCRA Subproblem} \label{SubSec:JCCRA}
For a given $\mathcal{N}_{\text{off}}$, i.e., $a_{n} = 1, \forall n \in \mathcal{N}_{\text{off}}$, the objective function of the JCCRA subproblem can be rewritten as 
\begin{equation} \label{Eq:JCCRA_New_ObjectiveFunction}
\min_{ \boldsymbol{P}, \boldsymbol{F} } \left( \sum\limits_{m \in \mathcal{M}}\sum\limits_{n \in \mathcal{N}_{m}}\frac{\lambda_{n}^{t}\alpha_{n} + \lambda_{n}^{e}u_{n}p_{n}}{R_{nm}} + \sum\limits_{m \in \mathcal{M}}\sum\limits_{n \in \mathcal{N}_{m}}\frac{\lambda_{n}^{t}\beta_{n}}{f_{nm}} \right),
\end{equation}
where $u_{n} = \alpha_{n} \left( \zeta_{n} \right)^{-1}$. Observe from the objective function~\eqref{Eq:JCCRA_New_ObjectiveFunction} that the first term is for the transmit power allocation of the mobile users and the second term is for the computation resource allocation of the MEC servers. In addition, the constraints C1, C5, and C6 are decoupled in $\boldsymbol{P}$ and $\boldsymbol{F}$. Therefore, it is possible to further decompose the JCCRA subproblem into two subproblems of $\boldsymbol{P}$ and $\boldsymbol{F}$. The two following subsections are devoted to the optimization of transmit power of the mobile users and the computation resource allocation of the MEC servers, respectively. 

\subsubsection{Transmit Power Allocation of Mobile Users}
We consider the optimization of the transmit power allocation $\boldsymbol{P}$ by solving the following problem (OPT-PA):
\begin{align} 
& \min_{ \boldsymbol{P} } \sum\limits_{m \in \mathcal{M}}\sum\limits_{n \in \mathcal{N}_{m}}\frac{\lambda_{n}^{t}\alpha_{n} + \lambda_{n}^{e}u_{n}p_{n}}{R_{nm}} \label{Eq:OPT_PA_ObjectiveFunction} \\
& \text{s.t.} \; 0 < p_{n}^{s} \leq p_{n}^{\max}, \forall n \in \mathcal{N}_{\text{pof}}, s \in \mathcal{S}.
\end{align}
Observe that the OPT-PA subproblem is a nonlinear fractional problem, which is highly complicated because of the existence of inter-cell interference among mobile users that offload to different MEC servers but on the same subchannel. With the offloading decision $\boldsymbol{A}^{*}$ from Algorithms~\ref{Alg:User_MECServer} and~\ref{Alg:User_Subchannels}, the problem~\eqref{Eq:OPT_PA_ObjectiveFunction} can be decomposed into $S$ subproblems, each subproblem can be written as 
\begin{align} 
& \min_{ \boldsymbol{P} } \sum\limits_{\left(m,n\right) \in \mathcal{G}_{s}}\frac{\lambda_{n}^{t}\alpha_{n} + \lambda_{n}^{e}u_{n}p_{n}^{s}}{R_{nm}^{s}} \label{Eq:OPT_PA_Decomposed_ObjectiveFunction} \\
& \text{s.t.} \: 0 < p_{n}^{s} \leq p_{n}^{\max}, \forall n \in \mathcal{G}_{s},
\end{align}
where $\mathcal{G}_{s} = \lbrace \left( m, n \right) \mid n \in \Psi(m), m = \Psi(n), \forall n \in \mathcal{N}_{\text{off}}, m \in \mathcal{M}, s \in \mathcal{S} \rbrace$. In~\eqref{Eq:OPT_PA_ObjectiveFunction} and~\eqref{Eq:OPT_PA_Decomposed_ObjectiveFunction}, we only consider mobile users that offload to different MEC servers but on the same subchannel. Obviously,~\eqref{Eq:OPT_PA_Decomposed_ObjectiveFunction} is still a non-linear fractional and non-convex problem due to the sum-of-ratios form of the objective function and the existence of inter-cell interference $I_{nm}^{s} = \sum\nolimits_{\left( m',n' \right) \in \mathcal{G}_{s}, n' \neq n}p_{n'}^{s}h_{n'm'}^{s}$ among mobile users in $\mathcal{G}_{s}$. One potential approach to the sum-of-ratios problem~\eqref{Eq:OPT_PA_Decomposed_ObjectiveFunction} relies on its transformation into a parametric convex programming problem \cite{Pham2017Energy}. However, in this paper, we find an approximate upper bound\footnote{Upper bound of $I_{nm}^{s}$ is due to the minimization problem~\eqref{Eq:OPT_PA_Decomposed_ObjectiveFunction}.} of $I_{nm}^{s}$ such that~\eqref{Eq:OPT_PA_Decomposed_ObjectiveFunction} can be decomposed into individual subproblems for different offloading users.

Suppose that the transmit power of mobile user $n$ is obtained by solving the following problem:
\begin{align} 
& \min_{ p_{n}^{s} } \frac{\lambda_{n}^{t}\alpha_{n}B_{s}^{-1} + \lambda_{n}^{e}u_{n}B_{s}^{-1}p_{n}^{s}}{\log_{2}\left( 1 + \frac{p_{n}^{s}h_{nm}^{s}}{n_{0} + I_{nm}^{s,0}} \right)} \label{Eq:OPT_PA_Individual_ObjectiveFunction} \\
& \text{s.t.} \: 0 < p_{n}^{s} \leq p_{n}^{\max},
\end{align}
where $I_{nm}^{s,0} =  \sum\nolimits_{\left(m',n'\right) \in \mathcal{G}_{s}, n' \neq n}p_{n'}^{\max}h_{n'm'}^{s}$. The problem~\eqref{Eq:OPT_PA_Individual_ObjectiveFunction} is still not easy to solve due to the fractional form of the objective function. In the following however, we show that the objective function of~\eqref{Eq:OPT_PA_Individual_ObjectiveFunction} is quasi-convex and the solution to~\eqref{Eq:OPT_PA_Individual_ObjectiveFunction} can be achieved using a bisection algorithm.
\begin{theorem}
The objective function of~\eqref{Eq:OPT_PA_Individual_ObjectiveFunction} is quasiconvex.
\end{theorem}
\begin{proof}
Let $\eta_{n}(p_{n}^{s}) = \frac{\lambda_{n}^{t}\alpha_{n}B_{s}^{-1} + \lambda_{n}^{e}u_{n}B_{s}^{-1}p_{n}^{s}}{\log_{2}\left( 1 + \frac{p_{n}^{s}h_{nm}^{s}}{n_{0} + I_{nm}^{s,0}} \right)}$, which is the ratio of a linear function and a concave function, and its sublevel sets $S_{a} = \lbrace p_{n}^{s} \in \left( 0 \phantom{c} p_{n}^{\max}\right] \mid \eta_{n}(p_{n}^{s}) \leq a \rbrace, \forall a \in \mathbf{R}^{+}$. The set $S_{a}$ can be equally expressed as 
\begin{align}
& S_{a} = \lbrace p_{n}^{s} \in \left( 0 \phantom{c} p_{n}^{\max}\right] \mid \lambda_{n}^{t}\alpha_{n}B_{s}^{-1} + \lambda_{n}^{e}u_{n}B_{s}^{-1}p_{n}^{s} \\
& \qquad\qquad - a \log_{2}\left( 1 + \frac{p_{n}^{s}h_{nm}^{s}}{n_{0} + I_{nm}^{s,0}} \right) \leq 0\rbrace, \forall a \in \mathbf{R}^{+}.
\end{align}
Let $f_{n}(p_{n}^{s}) = \lambda_{n}^{t}\alpha_{n}B_{s}^{-1} + \lambda_{n}^{e}u_{n}B_{s}^{-1}p_{n}^{s} - a \log_{2}\left( 1 + \frac{p_{n}^{s}h_{nm}^{s}}{n_{0} + I_{nm}^{s,0}} \right)$. According to \cite{Boyd2004Convex}, $S_{a}$ is a convex set if for any $\rho_{1}, \rho_{2} \in S_{a}$ and any $\theta$ with $0 \leq \theta \leq 1$, we have
\begin{equation} \label{Eq:Condition_Convex}
\theta \rho_{1} + \left( 1 - \theta \right)\rho_{2} \in S_{a}.
\end{equation}
The condition~\eqref{Eq:Condition_Convex} holds when $f_{n}\left( \theta \rho_{1} + \left( 1 - \theta \right)\rho_{2} \right) \leq 0$. Actually, $f_{n}(p_{n}^{s})$ is a convex function due to the subtraction of a linear function and a concave function. By definition, $f_{n}\left( \theta \rho_{1} + \left( 1 - \theta \right)\rho_{2} \right) \leq \theta f_{n}(\rho_{1}) + \left( 1 - \theta \right) f_{n}(\rho_{2})$. Due to $\rho_{1},\rho_{2} \in S_{a}$, we have $f_{n}(\rho_{1}) \leq 0$ and $f_{n}(\rho_{2}) \leq 0$. Therefore, the condition~\eqref{Eq:Condition_Convex} holds, $S_{a}$ is a convex set, and then $\eta_{n}(p_{n}^{s})$ is a quasiconvex function.
\end{proof}

One approach to quasiconvex optimization is a bisection algorithm, which solves a convex feasibility problem at each step \cite{Boyd2004Convex}. However, for further reduction of complexity, we use the approach proposed in \cite{Lyu2017Multiuser} to solve the quasiconvex optimization problem~\eqref{Eq:OPT_PA_Individual_ObjectiveFunction}. The basic idea is that the optimal solution $\tilde{p}_{n}^{s,*}$ either lies at the border of the constraint or satisfies the constraint $\partial \eta_{n}\left( \tilde{p}_{n}^{s,*} \right)/\partial p_{n}^{s} = 0$. We have $ \eta_{n}'\left( p_{n}^{s} \right) = \phi_{n}(p_{n}^{s})/\left( \log_{2}\left( 1 + p_{n}^{s}h_{nm}^{s}/\left(n_{0} + I_{nm}^{s,0}\right) \right) \right)^{2}, $ where 
\begin{align} \label{Eq:bisection_phi}
\phi_{n}(p_{n}^{s}) = & \lambda_{n}^{e}u_{n}B_{s}^{-1}\log_{2}\left( 1 + \frac{p_{n}^{s}h_{nm}^{s}}{n_{0} + I_{nm}^{s,0}} \right) \\ & - \frac{h_{nm}^{s}}{\ln 2}\frac{\lambda_{n}^{t}\alpha_{n}B_{s}^{-1} + \lambda_{n}^{e}u_{n}B_{s}^{-1}p_{n}^{s}}{n_{0} + I_{nm}^{s,0} + p_{n}^{s}h_{nm}^{s}}.
\end{align}
Moreover, the first-order derivative of~\eqref{Eq:bisection_phi} is expressed as
\begin{equation} \label{Eq:bisection_phi_1stderivative}
\phi_{n}'(p_{n}^{s}) = \frac{1}{\ln 2} \frac{\lambda_{n}^{t}\alpha_{n}B_{s}^{-1} + \lambda_{n}^{e}u_{n}B_{s}^{-1}p_{n}^{s}}{\left( \frac{n_{0} + I_{nm}^{s,0}}{h_{nm}^{s}} + p_{n}^{s} \right)^{2}}.
\end{equation}
From~\eqref{Eq:bisection_phi} and~\eqref{Eq:bisection_phi_1stderivative}, we have $\phi_{n}(0) < 0$ and $\phi_{n}'\left( \tilde{p}_{n}^{s} \right) > 0, \forall \tilde{p}_{n}^{s} \in \left( 0 \phantom{c} p_{n}^{\max} \right]$, i.e., $\phi_{n}(\cdot)$ is a monotonically increasing function and diminishes at $\tilde{p}_{n}^{s} = 0$. Therefore, iteratively checking the condition $\phi_{n}(\tilde{p}_{n}^{s}) \leq 0$, we can design an efficient bisection method as in Alg.~\ref{Alg:bisection_powercontrol}.
\begin{algorithm}
\caption{Bisection method for the quasiconvex optimization problem.}\label{Alg:bisection_powercontrol}
\begin{algorithmic}[1]
\State \textbf{Initialization} 
\State Set the tolerance $\varepsilon$, $p_{n}^{s,\text{l}} = 0$, and $p_{n}^{s,\text{u}} = p_{n}^{\max}$.
\State Compute $\phi_{n}\left( p_{n}^{s,\text{u}} \right)$.
\State \textbf{Find the optimal solution $\tilde{p}_{n}^{s,*}$}
\If {$\phi_{n}\left( p_{n}^{s,\text{u}} \right) \leq 0$}
\State $\tilde{p}_{n}^{s,*} = p_{n}^{s,\text{u}}$.
\Else
\Repeat
\State Set $p_{n}^{s,\text{t}} = \left( p_{n}^{s,\text{u}} + p_{n}^{s,\text{l}} \right) / 2$.
\If {$\phi_{n}\left( p_{n}^{s,\text{t}} \right) \leq 0$}
\State $p_{n}^{s,\text{l}} = p_{n}^{s,\text{t}}$.
\Else
\State $p_{n}^{s,\text{u}} = p_{n}^{s,\text{t}}$.
\EndIf
\Until{$p_{n}^{s,\text{u}} - p_{n}^{s,\text{l}} \leq \varepsilon$}
\State Set $\tilde{p}_{n}^{s,*} = \left( p_{n}^{s,\text{u}} + p_{n}^{s,\text{l}} \right) / 2$.
\EndIf
\State \textbf{Output}: the optimal solution $\tilde{p}_{n}^{s,*}$.
\end{algorithmic}
\end{algorithm}

We start the algorithm by introducing an upper-bound $p_{n}^{s,\text{u}}$ and a lower-bound $p_{n}^{s,\text{l}}$ of the transmit power $\tilde{p}_{n}^{s}$ and checking $\phi_{n}(\cdot)$ at the border of the constraint, i.e., $\tilde{p}_{n}^{s} = p_{n}^{\max}$. In each step, the interval is bisected, i.e., $p_{n}^{s,\text{t}} = \left( p_{n}^{s,\text{u}} + p_{n}^{s,\text{l}} \right) / 2$; therefore, the number of iterations required for Alg.~\ref{Alg:bisection_powercontrol} to terminate is $\lceil \log_{2}\left( p_{n}^{s,\text{u}} - p_{n}^{s,\text{l}} \right) / \varepsilon\rceil$. 

Note that the output of Alg.~\ref{Alg:bisection_powercontrol} is the approximate solution $\tilde{p}_{n}^{s,*}$ to the quasiconvex problem~\eqref{Eq:OPT_PA_Individual_ObjectiveFunction}. After finding the optimal solution to~\eqref{Eq:OPT_PA_Individual_ObjectiveFunction} for all mobile users in $\mathcal{G}_{s}$, $I_{nm}^{s}$ can be approximated as $ \tilde{I}_{nm}^{s} = \sum\nolimits_{n' \in \mathcal{G}_{s}, n' \neq n}\tilde{p}_{n'}^{s,*}h_{n'm}^{s}. $ Then, replacing $I_{nm}^{s,0}$ in~\eqref{Eq:OPT_PA_Individual_ObjectiveFunction} with $\tilde{I}_{nm}^{s}$, we obtain approximation problems for the power allocation of mobile users. The transmit power of mobile users is finally achieved by solving the approximation problems via the bisection Alg.~\ref{Alg:bisection_powercontrol}.

\subsubsection{Computation Resource Allocation of MEC servers}
The computation resource allocation $\boldsymbol{F}$ is determined by solving the following optimization problem (OPT-CRA): 
\begin{align} 
& \min_{ \boldsymbol{F} } \sum\limits_{m \in \mathcal{M}}\sum\limits_{n \in \mathcal{N}_{m}}\frac{\lambda_{n}^{t}\beta_{n}}{f_{nm}}  \label{Eq:OPT_CRA_ObjectiveFunction} \\
& \text{s.t.} 
\:\; f_{nm} > 0, \forall n \in \mathcal{N}_{m}, \forall m \in \mathcal{M} \label{Eq:OPT_CRA_Constr_ComputingVariable} \\
& 
\quad\;\; \sum\limits_{n \in \mathcal{N}_{m}}f_{nm} \leq f_{m}^{\max}, \forall m \in \mathcal{M}. \label{Eq:OPT_CRA_Constr_ComputingCapability}
\end{align}
The problem~\eqref{Eq:OPT_CRA_Constr_ComputingVariable} can be decomposed into $M$ individual problems, corresponding to $M$ MEC servers. However, even with~\eqref{Eq:OPT_CRA_Constr_ComputingVariable}, we will show that the optimal computation allocation of a single MEC server merely depends on the set of mobile users offloading to that MEC server. We have the following theorem on convexity of the OPT-CRA subproblem. 
\begin{theorem} \label{Theo:OPT_CRA_Convexity}
The OPT-CRA problem is a convex problem.
\end{theorem}
\begin{proof}
It is clear that the feasible solution set of the OPT-CRA is convex. The remaining task is to show the convexity of the objective function. We have the following derivatives
\begin{align}
& 
\frac{\partial^{2}}{\partial f_{nm}^{2}} \left( \sum\limits_{m \in \mathcal{M}}\sum\limits_{n \in \mathcal{N}_{m}}\frac{\lambda_{n}^{t}\beta_{n}}{f_{nm}} \right) = \frac{2\lambda_{n}^{t}\beta_{n}}{f_{nm}^{3}}, \forall n \in \mathcal{N}_{\text{off}}, m \in \mathcal{M} \\
&
\frac{\partial^{2}}{\partial f_{nm} \partial f_{kj}} \left( \sum\limits_{m \in \mathcal{M}}\sum\limits_{n \in \mathcal{N}_{m}}\frac{\lambda_{n}^{t}\beta_{n}}{f_{nm}} \right) = 0, \forall (m,m) \neq (k,j).
\end{align}
Let $\nabla^{2} g_{l}(\boldsymbol{F})$ be the Hessian matrix. Then, with all $\boldsymbol{v} \in \mathbb{R}^{N_{\text{off}}}$, $ \boldsymbol{v}^{T} \nabla^{2} g_{l}(\boldsymbol{F}) \boldsymbol{v} = \sum\nolimits_{n \in \mathcal{N}_{\text{off}}}2v_{n}^{2} \lambda_{n}^{t}\beta_{n}/f_{nm}^{3} \geq 0, $
where the equality happens if and only if $ \lambda_{n}^{t} = 0 $, i.e., the mobile user $ n $ is with an energy-hungry application. Therefore, the Hessian matrix is a positive semidefinite matrix. We conclude that OPT-CRA is a convex optimization problem.
\end{proof}

Since OPT-CRA is a convex problem, the optimal solution can be optimally achieved via the duality approach. Let $\boldsymbol{\nu} = \left\lbrace \nu_{m} \right\rbrace_{m \in \mathcal{M}}$ be the dual vector associated with the second constraint. The Lagrange function is given as
\begin{align}
& L_{\text{OPT-CRA}}(\boldsymbol{F}, \boldsymbol{\nu}) \\
& = \sum\limits_{m \in \mathcal{M}}\sum\limits_{n \in \mathcal{N}_{m}} \frac{\lambda_{n}^{t}\beta_{n}}{f_{nm}} + \sum\limits_{m \in \mathcal{M}}\nu_{m}\left( \sum\limits_{n \in \mathcal{N}_{m}}f_{nm} - f_{m}^{\max} \right).
\end{align}
Then, we define the Lagrange dual function $G_{\text{OPT-CRA}}(\boldsymbol{\nu})$ as
\begin{equation} \label{Eq:OPT_CRA_DualFunction}
G_{\text{OPT-CRA}}(\boldsymbol{\nu}) = \min_{\boldsymbol{F} \succ 0} L_{\text{OPT-CRA}}(\boldsymbol{F}, \boldsymbol{\nu}),
\end{equation}
which can be specified as the minimum of the Lagrangian function over the primal vector $\boldsymbol{F}$. Accordingly, this leads to the dual problem $ \max_{\boldsymbol{\nu} \succ 0} G_{\text{OPT-CRA}}(\boldsymbol{\nu}). $

Since the OPT-CRA subproblem is convex, the optimal computation resource $f_{nm}^{*}$ can be achieved by taking the first-order derivative of the Lagrange function $L_{\text{OPT-CRA}}(\boldsymbol{F}, \boldsymbol{\nu})$ with respect to (w.r.t.) $f_{nm}$ and setting the result equal to zero. Accordingly, we have 
\begin{equation} \label{Eq:OPT_CRA_Primal_Variable}
f_{nm}^{*} = \sqrt{\lambda_{n}^{t}\beta_{n}/\nu_{m}}.
\end{equation}
Replacing~\eqref{Eq:OPT_CRA_Primal_Variable} in~\eqref{Eq:OPT_CRA_DualFunction}, we can obtain the dual problem in $ \boldsymbol{\nu} $. This dual problem is also convex, and the optimal dual vector $\nu_{m}^{*}$ is therefore achieved by setting the first-order derivative of $G_{\text{OPT-CRA}}(\boldsymbol{\nu})$ w.r.t. $\nu_{m}^{*}$ equal to zero. We then have $ \nu_{m}^{*} = \left( \sum\nolimits_{n \in \mathcal{N}_{m}}\sqrt{\lambda_{n}^{t}\beta_{n}}/f_{m}^{\max} \right)^{2}. $
Now, substituting $ \nu_{m}^{*} $ back into~\eqref{Eq:OPT_CRA_Primal_Variable}, the optimal computation resource is obtained, as follows: 
\begin{equation}  \label{Eq:OPT_CRA_Optimal_Computation_Resource}
f_{nm}^{*} = \frac{f_{m}^{\max}\sqrt{\lambda_{n}^{t}\beta_{n}}}{\sum\nolimits_{n \in \mathcal{N}_{m}}\sqrt{\lambda_{n}^{t}\beta_{n}}}.
\end{equation}
\begin{remark}
It is revealed from~\eqref{Eq:OPT_CRA_Optimal_Computation_Resource} that the computation resource is determined by the weighted parameter $\lambda_{n}^{t}$, the number of CPU cycles required to complete the task of all mobile users. The weighted parameter, $\lambda_{n}^{t}$, can be interpreted as the importance level of computational time of the mobile user $n$. Specifically, if all mobile users have the same computation task requirement, i.e., $\beta_{n} = \beta_{k}, \forall n \neq k, n, k \in \mathcal{N}_{m}$, the larger the value of the weighted parameter $\lambda_{n}^{t}$ is, the more the computation resource should be assigned to the mobile user $n$ by the corresponding MEC server $m$ in order to minimize the processing time.
\end{remark}

\subsection{Algorithm for JCORA Problem} \label{SubSec:JointAlgorithm}
In this subsection, we propose a joint framework to find the optimal solution to the underlying problem~\eqref{Eq:OPT_Constr_UE_OffloadingDecision}. The specific details of the proposed algorithm are summarized in Alg.~\ref{Alg:JointAlgorithm}, which is referred to as JCORAMS (JCORA Multi-Server). In general, the proposed algorithm is composed of three phases: the pre-computation offloading decision, computation offloading and resource allocation, and post-computation offloading decision. The purpose of the first phase is to filter out mobile users who cannot benefit from computation offloading, i.e., those users who should execute their tasks locally, and to reduce the input dimension for the second phase, i.e., non-offloading users are not taken into account during the second phase.

\begin{algorithm}
\caption{JCORAMS algorithm.}\label{Alg:JointAlgorithm}
\begin{algorithmic}[1]
\State \textbf{Input}: $\mathcal{M}$, $\mathcal{N}$, $\mathcal{S}$. 
\State \textbf{Pre-computation offloading decision} 
\State \quad\; Each user $n$ decides its minimal offloading gain $\Upsilon_{n}$.
\State \quad\; Each user checks the condition~\eqref{Eq:Condition_Offloading} to determine the \phantom{ccic}offloading decision.
\State \textbf{Find the optimal solution}
\State \textit{Matching between mobile users and MEC servers}
\State \quad\; Each user calculates its preference according to~\eqref{Eq:Preference_UA_MobileUser}.
\State \quad\; Each MEC server constructs its preference via~\eqref{Eq:Preference_UA_MECServer}.
\State \quad\; Obtain the optimal matching $\Psi^{*}$ via Alg.~\ref{Alg:User_MECServer}.
\State \textit{Matching for subchannel allocation in a single SeNB $m$}
\For {$m = 1$ to $M$}
\State Each user has its preference over $S$ subchannels.
\State Each subchannel gets its preferences over $N_{m}$ users.
\State Obtain the optimal matching $\Omega^{*}$ via Alg.~\ref{Alg:User_Subchannels}.
\EndFor
\State \textit{Transmit power allocation of mobile users}
\For {$s = 1$ to $S$}
\State Each user in $\mathcal{G}_{s}$ finds $\tilde{p}_{n}^{s,*}$ and $\tilde{I}_{nm}^{s}$.
\State Solve~\eqref{Eq:OPT_PA_Individual_ObjectiveFunction} with $I_{nm}^{s,0} = \tilde{I}_{nm}^{s}$ for mobile user $n$. 
\EndFor
\State \textit{Computation resource allocation of MEC server}
\For {$n = 1$ to $M$}
\State SeNB $m$ collects $f_{n}^{l}$ and $\beta_{n}$ from all associated users.
\State Computation resource is allocated according to~\eqref{Eq:OPT_CRA_Optimal_Computation_Resource}.
\EndFor
\State \textbf{Post-computation offloading decision} 
\For {$n = 1$ to $\vert \mathcal{N}_{\text{pof}} \vert$}
\State Compute the value with $m \in \Psi^{*}(n)$ and $s \in \Omega^{*}(n)$
\begin{equation}
\Upsilon_{n}^{*} = \frac{\lambda_{n}^{t} \alpha_{n} + \lambda_{n}^{e}p_{n}^{*}\alpha_{n}\zeta_{n}^{-1}}{R_{nm^{*}}^{s^{*}}} + \frac{\lambda_{n}^{t}\beta_{n}}{f_{nm^{*}}} - Z_{n}^{l}.
\end{equation}
\If {$\Upsilon_{n}^{*} > 0$}
\State Update $a_{nm^{*}} = 0$ and $a_{nm^{*}}^{s^{*}} = 0$.
\EndIf
\EndFor
\State Repeat the second phase until $\Upsilon_{n}^{*} \leq 0, \forall n \in \mathcal{N}_{\text{pof}}$.
\State \textbf{Output}: the optimal solution $\left( \boldsymbol{A}^{*}, \boldsymbol{P}^{*}, \boldsymbol{F}^{*} \right)$.
\end{algorithmic}
\end{algorithm}

The second phase is further divided into four steps: user-server association, subchannel allocation, transmit power control, and computation resource allocation. 
\begin{itemize}
\item \textit{User-server association}: The mobile users and MEC servers join a one-to-many matching via Alg.~\ref{Alg:User_MECServer}. The preference of a mobile user over potential MEC servers and the preference of an MEC server over the $|\mathcal{N}_{\text{pof}}|$ users are calculated according to~\eqref{Eq:Preference_UA_MobileUser} and~\eqref{Eq:Preference_UA_MECServer}, respectively. The optimal user association is obtained via Alg.~\ref{Alg:User_MECServer}, where a user sends the proposal to offload its task to the most preferred MEC server and an MEC server accepts a number of preferred mobile users based on its quota. Alg.~\ref{Alg:User_MECServer} stops when every user is either accepted by one MEC server or rejected by all preferred MEC servers.

\item \textit{Subchannel allocation}:  After Alg.~\ref{Alg:User_MECServer}, all mobile users which know their associated MEC servers and all mobile users offloading to the same MEC server join a one-to-one matching game. Each user in $\mathcal{N}_{m}$ calculates its preference over all subchannels according to~\eqref{Eq:Preference_CA_User} and the MEC server $m$ computes its preference on all subchannels over its associated users according to~\eqref{Eq:Preference_CA_Subchannel}. A user sends the proposal to the most preferred subchannel and a SeNB assigns a subchannel to the most preferred user, who has the highest preference among requested users, and rejects the proposals of other mobile users on that subchannel. Alg.~\ref{Alg:User_Subchannels} terminates when there is no bidding between mobile users and subchannels.

\item \textit{Transmit power control}: Once two matching algorithms for user association and subchannel allocation terminate, the transmit power of offloading users is allocated. Note that there are $S$ groups $\mathcal{G}_{s}$ and the transmit power of mobile users in the group $\mathcal{G}_{s}$ is achieved by solving the individual problem~\eqref{Eq:OPT_PA_Individual_ObjectiveFunction}. The approximate transmit power of a user in $\mathcal{G}_{s}$ is found via Alg.~\ref{Alg:bisection_powercontrol} by fixing the inter-cell interference at the maximal transmit power of the other users. After that, the inter-cell interference of mobile users can be well approximated, and these approximation problems are solved to find the optimal transmit power of mobile users in $\mathcal{G}_{s}$.

\item \textit{Computation resource allocation}: Computation resource allocation at MEC servers is executed when Alg.~\ref{Alg:User_Subchannels} terminates. Each MEC server $m$ allocates the computation resources to its associated users in $\mathcal{N}_{m}$ according to~\eqref{Eq:OPT_CRA_Optimal_Computation_Resource}.  
\end{itemize} 

The third phase acts as the second filter since we assume that all of the mobile users in $\mathcal{N}_{\text{pof}}$ offload their tasks to the MEC servers. After each iteration $ t $, it is necessary to determine whether or not the mobile users benefit from computation offloading with resource allocation from the second phase (line 28). If mobile users still do not benefit from computation offloading, they are possibly removed from the set of offloading users (line 30). Here, among those users, one with the lowest local computation overhead is selected and removed from the set of offloading users. The proposed algorithm converges and terminates when the matching of two consecutive iterations $ t $ remains unchanged (line 33). 

\subsection{Convergence and Stability} \label{SubSec:Convergence_Stability}
In order to analyze the convergence and stability of the proposed algorithm, let us consider the group $\mathcal{G}_{s}, \forall s \in \mathcal{S}$ and introduce the definition of group stable \cite{LeAnh2017Matching, Bayat2014Distributed}.  
\begin{definition} [Group Stable] \label{Def:GroupStable}
The group $\mathcal{G}_{s}$ is blocked by a group $\mathcal{G}_{s}', \forall s \in \mathcal{S}$, which comprises of at least one MEC server and one mobile user, if there exists another matching $\Psi'$ such that $\forall n,m \in \mathcal{G}_{s}'$,
\begin{enumerate}
\item $\phi_{n,\text{UA}}\left( \Psi'(n) \right) > \phi_{n,\text{UA}}\left( \Psi(n) \right)$,
\item $\phi_{m,\text{UA}}\left( \Psi'(m) \right) > \phi_{m,\text{UA}}\left( \Psi(m) \right)$.
\end{enumerate}
The group $\mathcal{G}_{s}$ is said to be group stable if it is not blocked by any group. In addition, matchings in the proposed algorithm are stable if and only if all groups $\mathcal{G}_{s}, \forall s \in \mathcal{S}$ are group stable. 
\end{definition}
In Definition~\ref{Def:GroupStable}, the first and second conditions express that all MEC servers and mobile users in $\mathcal{G}_{s}'$ prefer their matches in $\Psi'$ to their current matches in $\Psi$. In other words, the group $\mathcal{G}_{s}$ is blocked by a group $\mathcal{G}_{s}'$ if all MEC servers and mobile users in $\mathcal{G}_{s}'$ find a more preferable matching than their current matchings.

\begin{theorem}
Matchings $\mathcal{G}_{s}, \forall s \in \mathcal{S}$ generated by Algorithms~\ref{Alg:User_MECServer} and~\ref{Alg:User_Subchannels} are stable in each iteration of the proposed Alg.~\ref{Alg:JointAlgorithm}.
\end{theorem}
\begin{proof}
The proof is similar to that in Appendix A in \cite{LeAnh2017Matching}. It is therefore omitted.
\end{proof}

\begin{theorem}
The proposed algorithm generates a group stable $\mathcal{G}_{s}$ after a finite number of iterations and is guaranteed to converge. 
\end{theorem}
\begin{proof}
The number of preference relations of the mobile users, MEC servers, and subchannels, i.e., $\succ_{n,\text{UA}}$, $\succ_{m,\text{UA}}$, $\succ_{n,\text{CA}}$, and $\succ_{s,\text{CA}}$, in each iteration is finite since the numbers of mobile users, MEC servers, and subchannels is finite. Additionally, matchings in each iteration are proved to be stable and the number of preference relations reduces after each iteration. Therefore, the group $\mathcal{G}_{s}, \forall s \in \mathcal{S}$, generated by the user association and channel allocation phases of the proposed algorithm, are all stable. Moreover, the transmit power of the offloading users and the computation resources at the MEC servers are derived based on the simple approximation approach and convex technique. As a result, the proposed algorithm is guaranteed to converge.
\end{proof}

The optimality property of a stable matching can be observed by \textit{weak Pareto optimality} (PO) \cite{Manlove2013Algorithmics}.  Denote by $ Z\left(\mathcal{G}\right) $ the total computation overhead obtained by the matching $ \mathcal{G} $ and the corresponding $ \left( \boldsymbol{P}, \boldsymbol{F} \right)$, where  $ \mathcal{G} = \left(\Psi, \Omega\right) $. The matching $ \mathcal{G} $ is weak PO if there is no other matching $ \mathcal{G}' $ with $ Z\left(\mathcal{G}'\right) \preceq Z\left(\mathcal{G}\right) $, which is strict for one user \cite{Jorswieck2011Stable}. 
\begin{theorem}
	The JCORAMS algorithm procudes a weak PO solution to the underlying problem.
\end{theorem}
\begin{proof}
	Let us consider $ \mathcal{G} $ to be a stable matching obtained by Alg.~\ref{Alg:JointAlgorithm} and assume that there is a unstable matching $ \mathcal{G}' $, which is PO to $ \mathcal{G} $. There are two reasons behind instability of  $ \mathcal{G}' $; it is either 1) lack of individual rationality or 2) blocked by another matching. 
	
	For case 1, assume that the user $ n $ is not individually rational, its computation overhead can be reduced by matching $ n $ with another $ \mathcal{G}(n) $ instead of the currently matching $ \mathcal{G}'(n) $. This decreases the computation overhead of the user $ n $ and then $ Z\left(\mathcal{G}\right) < Z\left(\mathcal{G}'\right) $ since the computation overhead of other users is left unchanged in $ \mathcal{G}' $. For case 2, we assume that the unstable matching $ \mathcal{G}' $ is blocked by $ \left(n,m,s\right) $, i.e., the user $ n $, server $ m $, and subchannel $ s $. The second case happens when $ n $ strictly prefers $ \left(m,s\right) $ to $ \mathcal{G}'\left(n\right) $. We can construct a new stable matching $ \mathcal{G} $ by assigning $ n $ to $ \left(m,s\right) $ instead of  $ \mathcal{G}'\left(n\right) $. That leads to $ Z_{n}\left(\mathcal{G}\right) < Z_{n}\left(\mathcal{G}'\right)$, and then since the other preferences remain unchanged, $ U\left(\mathcal{G}\right) < U\left(\mathcal{G}'\right)$. From both cases, we conclude that there is no unstable matching that can generate smaller system-wide computation overhead compared to the stable matching $ \mathcal{G} $. As a consequence, the matching produced by Alg.~\ref{Alg:JointAlgorithm} is stable and weak PO. 
\end{proof}

\section{Numerical Simulation}
In this section, we present numerical simulations to evaluate the performance of the proposed algorithm. 
\label{Sec:Simulation}
\subsection{Simulation Settings}
In order to evaluate our proposed algorithm, we use the following simulation settings for all simulations. We first consider the scenario where $9$ SeNBs are randomly deployed in a small indoor area of $250 \times 250$ m$^{2}$ to serve 36 mobile users. Each SeNB consists of 4 subchannels and has a quota of 4 users ($q_{m} = 4, \forall m \in \mathcal{M}$) . The bandwidth of each subchannel is $B_{s} = 5$ MHz, each mobile user has the maximum transmit power $p_{n}^{\max} = 100$ mW, and $n_{0} = -100$ dBm. The path-loss model is $-140.7 - 36.7\log_{10}\left(d\right)$, where $d$ (kilometers) is the distance from the user to the serving SeNB. For the computation task, we adopt the face recognition application in \cite{Lyu2017Multiuser, Chen2016Efficient, Soyata2012Cloud}, where the computation input data size is 420 KB and the total required number of CPU cycles is 1000 Megacycles. The CPU computational capability $f_{n}^{l}$ of the mobile user $n$ is randomly assigned from the set $\lbrace 0.5, 0.8, 1.0 \rbrace$ GHz \cite{Chen2016Efficient, Chen2015Decentralized} and the computational capability of each MEC server $4.0$ GHz. The weighted parameters of the computational time and energy consumption are both $0.5$, i.e., $\lambda_{n}^{t} = \lambda_{n}^{e} = 0.5, \forall n \in \mathcal{N}$. Finally, we set the values of $\varphi_{\text{UA}}$, $\varepsilon_{\text{UA}}$, $\varphi_{\text{CA}}$, and $\delta_{m}^{s}$ ($\forall m \in \mathcal{M}, s \in \mathcal{S}$) to $8 \times 10^{6}$, $0.2$, $1$, and $0.1$, respectively. For all the results, each plot is the average of 100 channel realizations and in each realization, the mobile user and MEC server locations are uniformly distributed randomly.

\subsection{Simulation Results}
In the following, we will present the performance of our proposed approach compared with several representative benchmark methods. For existing frameworks, the following solutions are considered:
\begin{enumerate}
	\item \textit{Local computing only}: there is no computation offloading. All mobile users perform computations locally, i.e., $a_{n} = 0, \forall n \in \mathcal{N}$.
	
	\item \textit{Offloading only}: all mobile users offload their computation tasks to the MEC servers i.e., $a_{n} = 1, \forall n \in \mathcal{N}$. This is achieved by running our algorithm without the pre-computation offloading decision and post-computation offloading decision steps. Note that when the number of mobile users exceeds the system capacity, some requested users are rejected by the algorithm, i.e., $\exists n \in \mathcal{N} \vert a_{n} = 0$.
	
	\item \textit{HODA} \cite{Lyu2017Multiuser}: the offloading decision, transmit power, and computation resources are determined in each cell independently. 
\end{enumerate}
To allow fair comparison between algorithms for single MEC server and multiple MEC servers, the final results, i.e., percentage of offloading users and system-wide computation overhead, do not take into account the inter-cell interference among offloading users. 

\begin{figure}[!ht]
\centering
\includegraphics[width=0.46\linewidth]{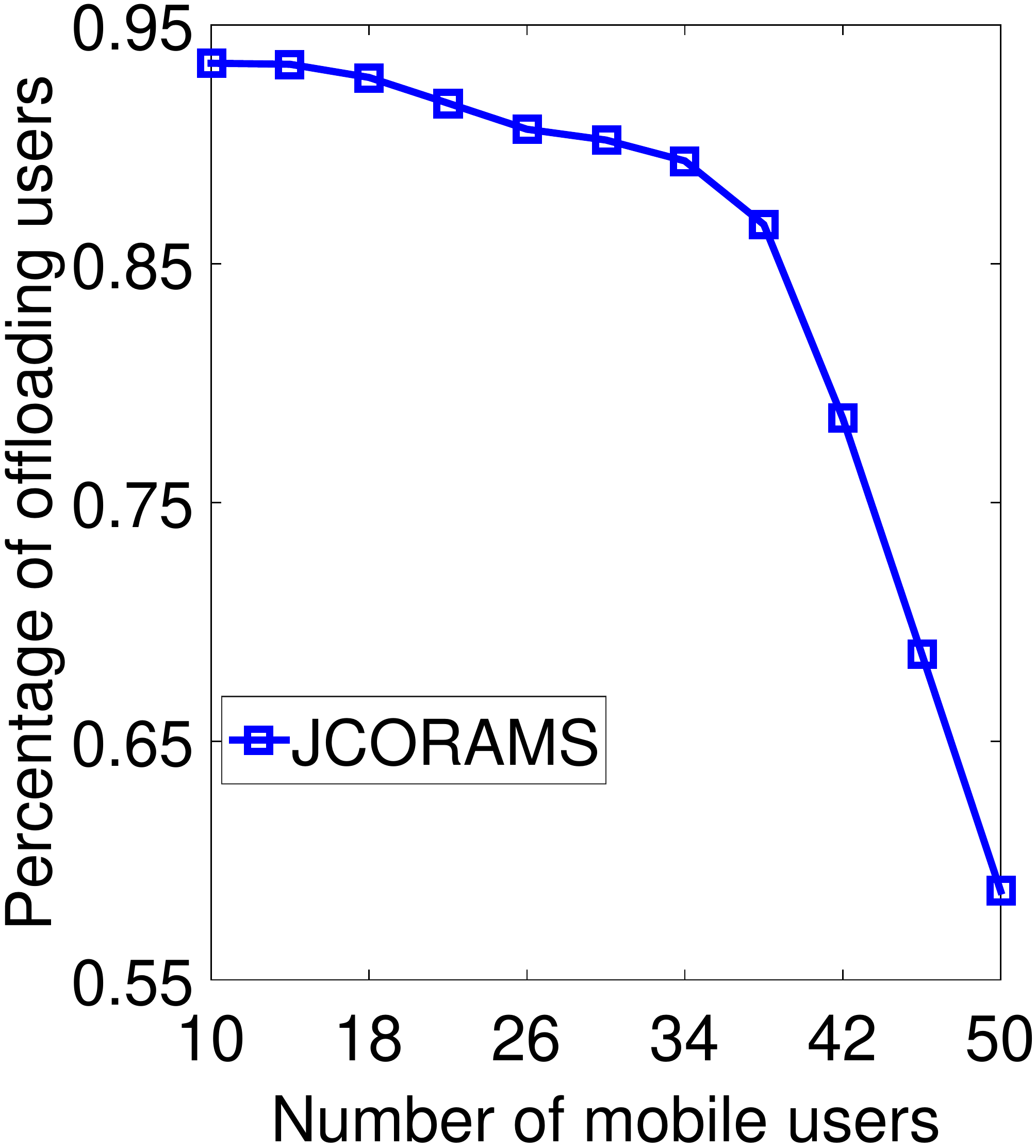}
\caption{Percentage of offloading users.}
\label{Fig:Percentage_of_OffloadingUsers}
\end{figure}
In the first experiment, we vary the number of mobile users from $10$ to $50$ with a step deviation of $ 4 $ and examine the percentage of offloading users. From Fig.~\ref{Fig:Percentage_of_OffloadingUsers}, the percentage of offloading users is relatively high\footnote{The percentage of offloading users should be $1$ when the number of mobile users is relatively small. However, mobile user and MEC server locations are both randomly generated in each simulation realization, so a user may not offload its computation task due to bad channel connections with MEC servers.} when the total number of mobile user is small. However, the percentage of offloading users gradually decreases when the total number of mobile users increases. This is reasonable since 1) each user might have a high probability to associate with its preferred MEC server and offload its computation task over a good subchannel and 2) small intercell interference makes mobile users profit more from computation offloading. In addition, when the number of mobile users keeps increasing, each mobile user needs to compete with the others for using radio resource and computation resource, and due to the limited number of MEC servers, number of subchannels in each cell, and quota of each MEC server, a portion of requested users must be rejected by the proposed algorithm.

\begin{figure}[!ht]
  \centering
    \subfloat[Percentage of offloading users.\label{Fig:Percentage_of_OffloadingUsers_vs_NoOfUsers_Comparison}]{\includegraphics[width=0.48\linewidth]{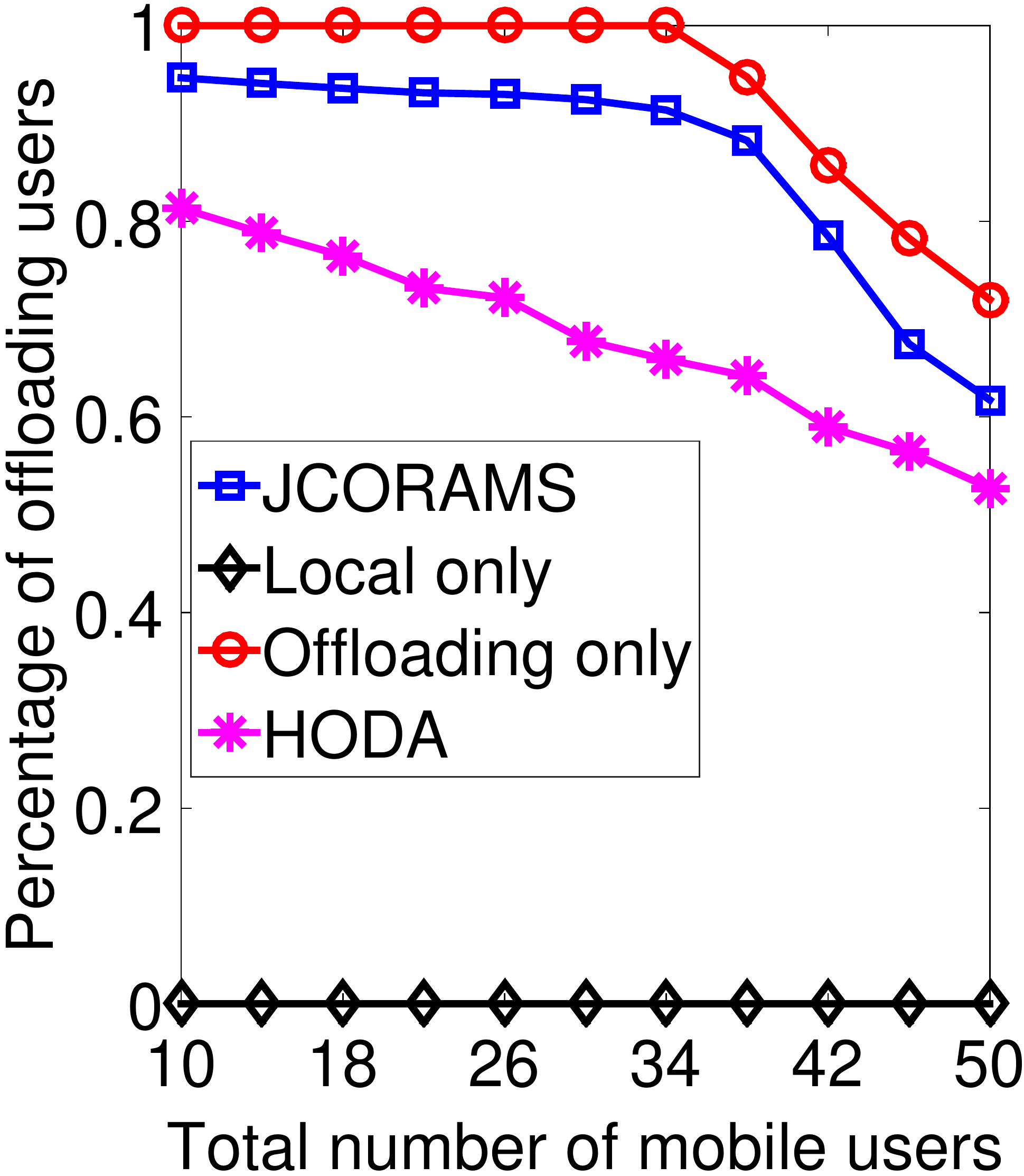}}\;
    \subfloat[Total computation overhead.\label{Fig:Computation_Overhead_vs_NoOfUsers_Comparison}]{\includegraphics[width=0.49\linewidth]{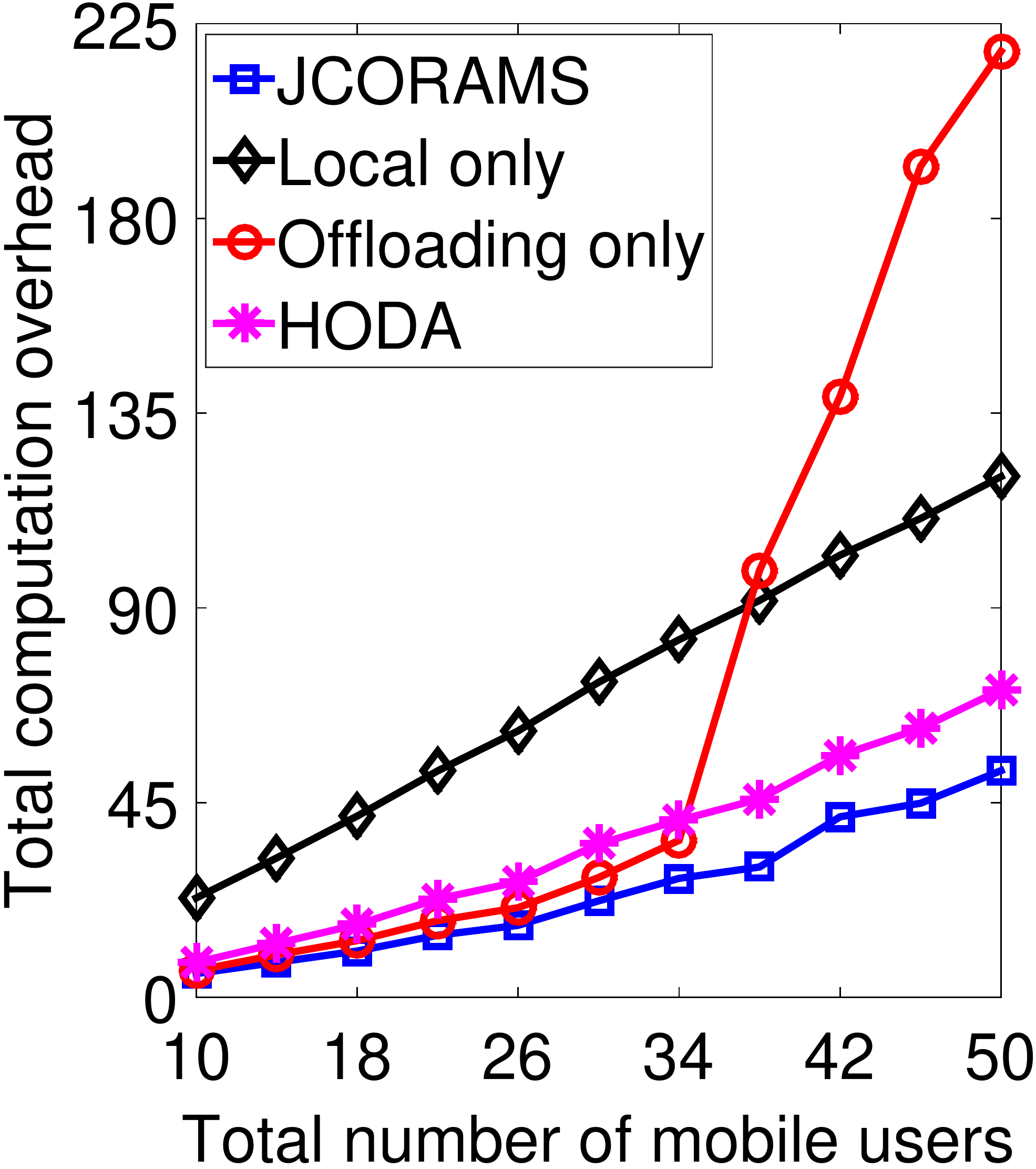}}
  \caption{Comparison of JCORAMS and three baseline frameworks under different numbers of mobile users.}
  \label{Fig:PerformanceComparison_vs_NoOfUsers}
\end{figure}
In the second experiment, we vary the number of mobile users and examine the performances in terms of the percentage of offloading users and system-wide computation overhead for our proposed approach and the above three existing frameworks. It is observed from Fig.~\ref{Fig:Percentage_of_OffloadingUsers_vs_NoOfUsers_Comparison} that the percentage of offloading users in the local computing only method is $0$ while that of the offloading only approach is $1$, which starts to decrease as the total number of mobile users become smaller and greater than $36$, respectively. This is due to the quota of each SeNB, the number of subchannels of each cell, and the number of MEC servers, and hence the system capacity in terms of the number of admitted offloading users is limited by $M\times\min\lbrace q, S  \rbrace$. As the total number of mobile users increases, the percentage of offloading users tends to decrease and the system-wide computation overhead increases. This can be explained as follows. First, the larger the total number of mobile users is, the lower the probabilities for each user to connect with its preferred MEC server and subchannel are, and the intercell interference among offloading users in different MEC servers becomes more severe. Second, only a fraction of mobile users are able to offload their computation tasks to the MEC servers while the remaining users do not benefit from computation offloading and thus must execute their tasks locally. Fig.~\ref{Fig:Computation_Overhead_vs_NoOfUsers_Comparison} also reveals that the performance of the offloading only algorithm becomes worse than that for the local computing only method when the total number of mobile users gets larger. This is due to competition among mobile users for the limited radio and computation resources. Compared with three baseline schemes, i.e., offloading only, local only, and HODA (for single MEC server), our proposed algorithm can achieve better performance in terms of the percentage of offloading users and yield a lower computation overhead. 

Similar to the second experiment, the third experiment compares the performances of our proposed algorithm with the existing alternative frameworks under different computation task profiles. It is shown in Figs.~\ref{Fig:Percentage_of_OffloadingUsers_vs_Alpha_Comparison}-\ref{Fig:Computation_Overhead_vs_Alpha_Comparison} that when the input data size $\alpha$ is large enough ($1$ MB in this case), the computation overhead of the offloading only method can reach that of the local computing only scheme. It is therefore better to offload fewer computation tasks to the MEC servers as the input data size $\alpha$ increases, i.e., a computation task with small data size is more preferable to computation offloading than one with high data size. The reason for this is that by increasing the input data size, the time cost and energy cost for offloading computation tasks become higher, as seen from Eqs.~\eqref{Eq:overhead_r_offloadingtime} and~\eqref{Eq:overhead_r_offloadingenergy}. This observation agrees with the performance lines of JCORAMS and HODA, where fewer users benefit from computation offloading and the system-wide computation overhead increases as the input data size $\alpha$ increases. From Figs.~\ref{Fig:Percentage_of_OffloadingUsers_vs_Beta_Comparison}-\ref{Fig:Computation_Overhead_vs_Beta_Comparison}, we can observe that the percentage of offloading users and the system-wide computation overhead increase with the number of CPU cycles $\beta$ required to accomplish the computation tasks. This is reasonable since both the local completion time and remote execution time increase as $\beta$ increases; however, the computational capability of a mobile user is often limited and an MEC server can offer offloading users with higher computational capability, i.e., mobile users therefore benefit from computation offloading if their computation tasks are executed by the MEC servers. From the above reasons, we conclude that it is better to offload a computation task with small input data size and large computation intensity rather than one with large input data size and small computation intensity. Obviously, the proposed algorithm achieves the better performance than the baseline solutions. 
\begin{figure}[!ht]
  \centering
    \subfloat[Percentage of offloading users.\label{Fig:Percentage_of_OffloadingUsers_vs_Alpha_Comparison}]{\includegraphics[width=0.48\linewidth]{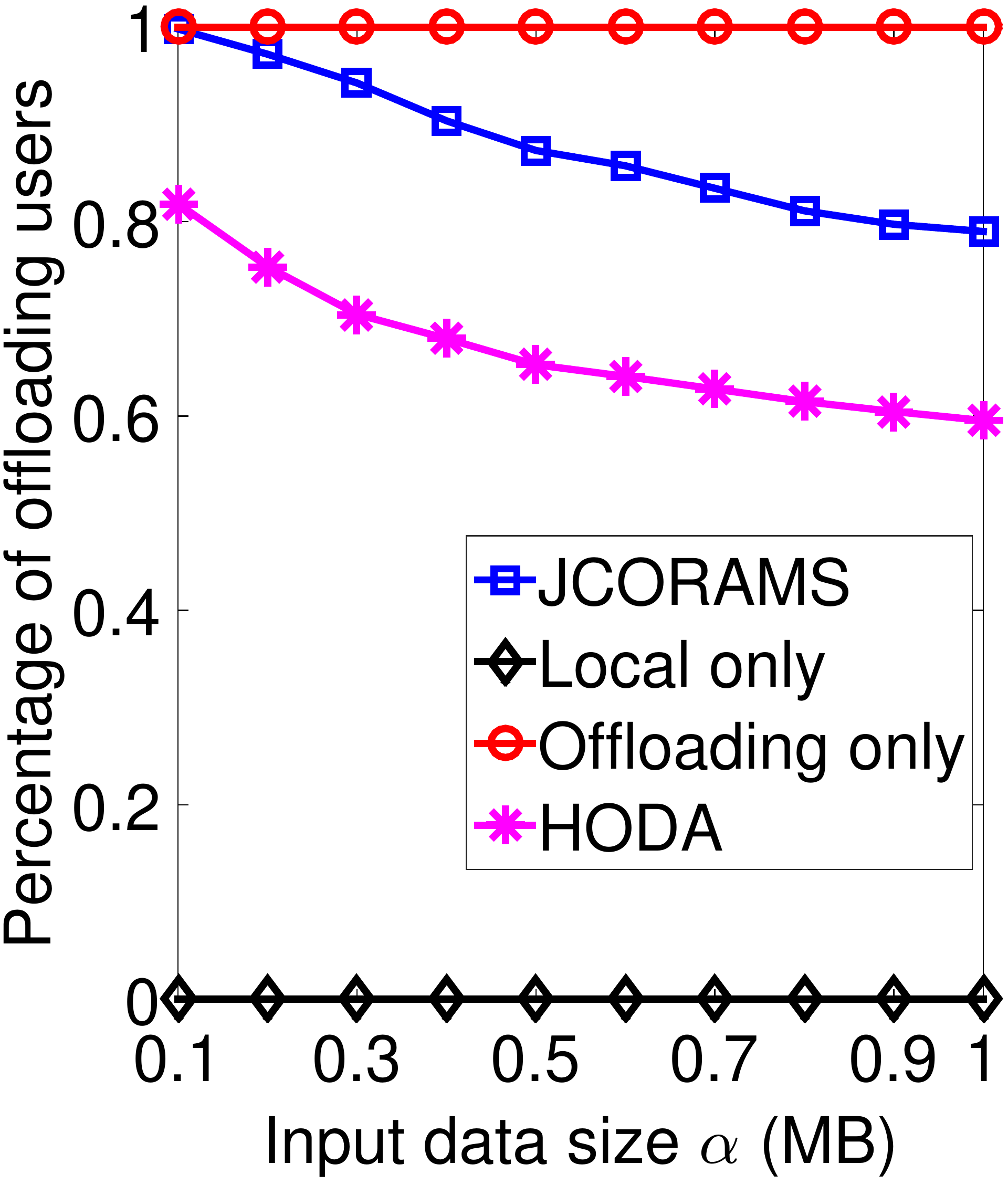}}\;
    \subfloat[Total computation overhead.\label{Fig:Computation_Overhead_vs_Alpha_Comparison}]{\includegraphics[width=0.475\linewidth]{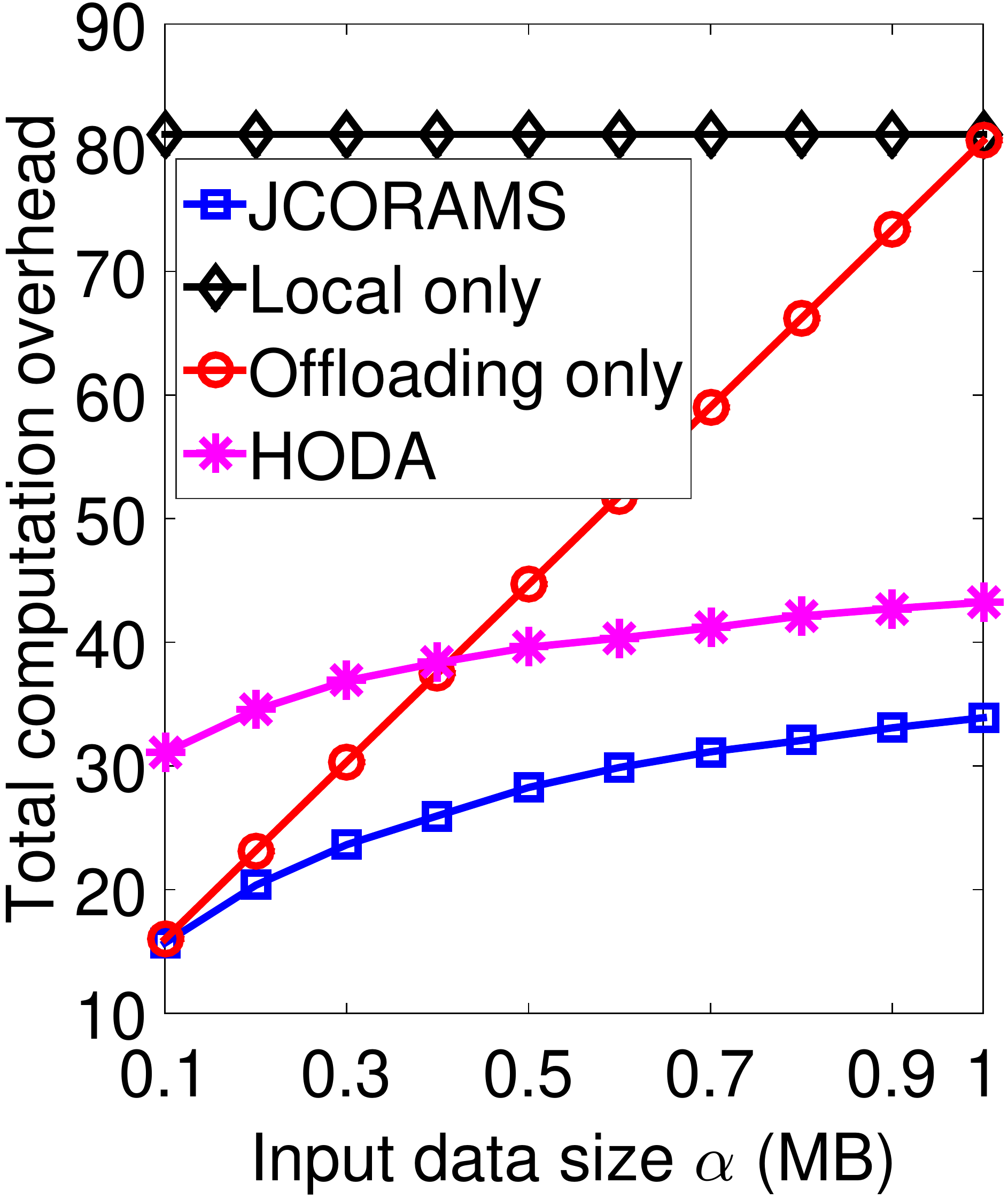}}\;
    \subfloat[Percentage of offloading users.\label{Fig:Percentage_of_OffloadingUsers_vs_Beta_Comparison}]{\includegraphics[width=0.48\linewidth]{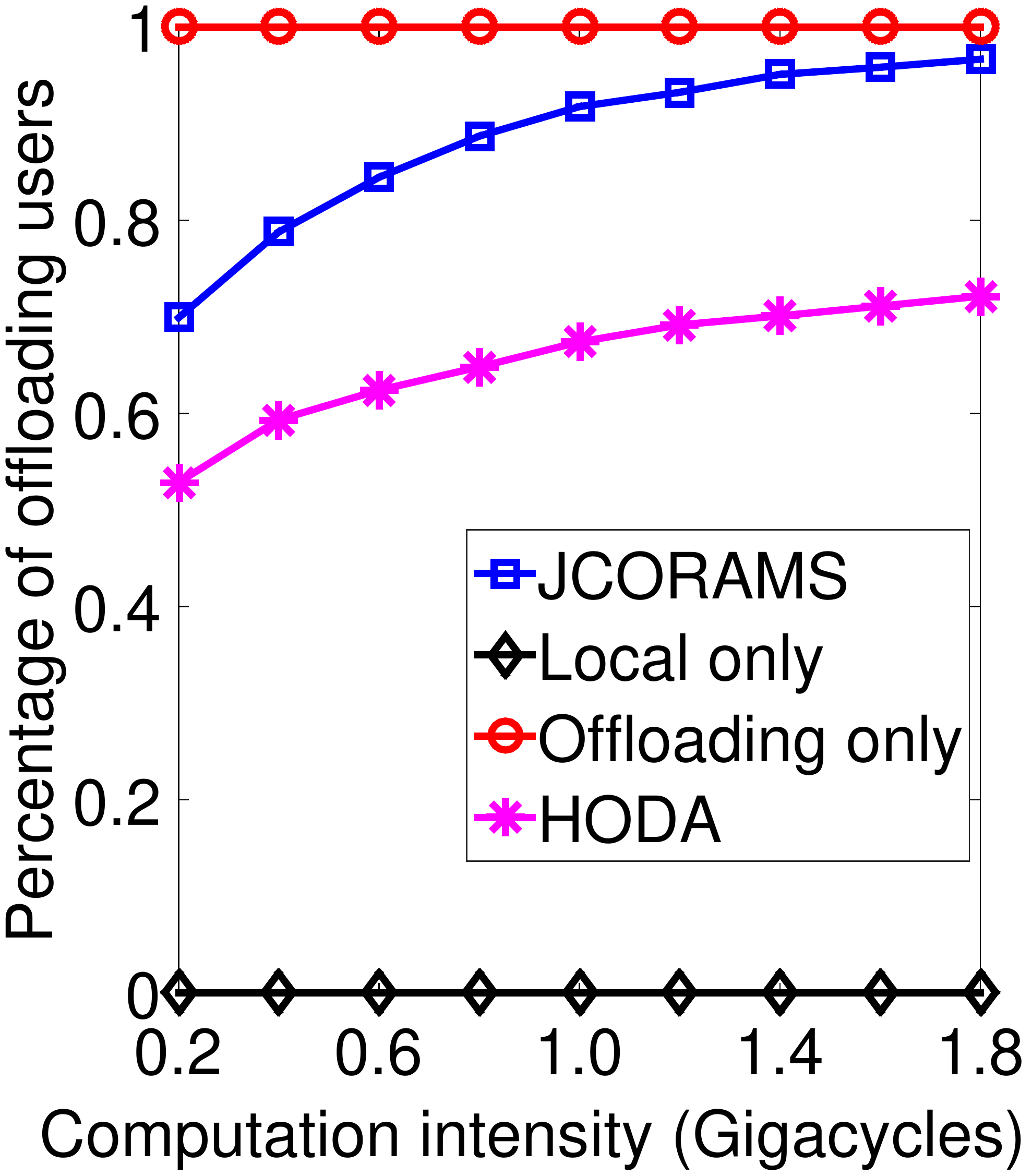}}\;
    \subfloat[Total computation overhead.\label{Fig:Computation_Overhead_vs_Beta_Comparison}]{\includegraphics[width=0.48\linewidth]{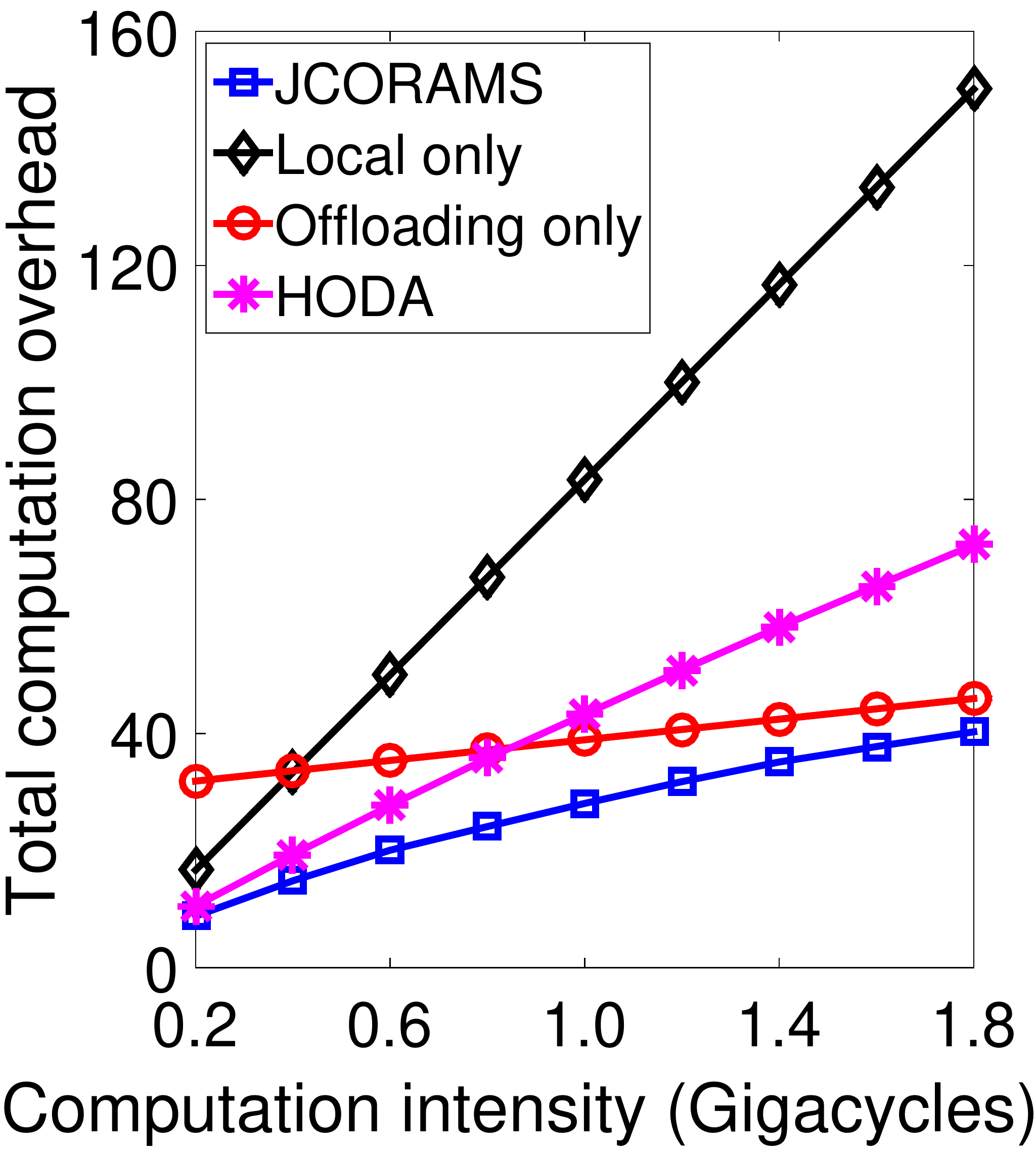}}
  \caption{Comparison of our proposed approach and three alternatives under different computation task profiles.}
  \label{Fig:PerformanceComparison_vs_AlphaBeta}
\end{figure}

Next, by varying the weighted parameter of the computational time $\lambda_{t}$ from $0.1$ to $0.9$ with a deviation step of $0.1$ and setting the weighted parameter of the energy consumption $\lambda_{e} $ to $1 - \lambda_{t}$, we further explore the performance comparison between our proposed algorithm and existing ones. It is worth noting that the weighted parameters are the same for all mobile users; however, the extension to different weighted parameters for different users does not affect the comparison among algorithms. Selecting a mobile user with $f_{l} = 0.8$ GHz as an example, we have $t^{l} = 1000 \times 10^{6}/0.8 \times 10^{9} = 1.25$ (seconds) and $E^{l} = 5 \times 10^{-27} \times 1000 \times 10^{6} \times \left( 0.8 \times 10^{9} \right)^{2} = 3.20$ (Joules). It is obvious that for local computing, the energy consumption is nearly three times larger than the completion time. As a result, with the increment of the weighted parameter of computation time $\lambda_{t}$ as well as the decrement of the weighted parameter of energy consumption $\lambda_{e}$, the system-wide computation overhead by the local only and offloading only schemes decreases and increases, respectively, as observed from Fig.~\ref{Fig:Computation_Overhead_vs_Weights_Comparison}. In addition, there are fewer users that tend to offload their computation tasks to the MEC servers due to the lower computation overhead from local computing, and the percentage of offloading users reduces, as shown in Fig.~\ref{Fig:Percentage_of_OffloadingUsers_vs_Weights_Comparison}. Here, the system-wide computation overheads by JCORAMS and HODA still increase and only start to decline when $\lambda_{t}$ is large enough. The main reason for this is the dominance of the offloading and execution time ($t^{\text{off}} + t^{\text{exe}}$) over the offloading energy consumption ($E^{\text{off}}$). Therefore, selection of the weighted parameters plays an important role in the achieved performances. Again, our proposed algorithm is superior to the baseline solutions.
\begin{figure}[!ht]
  \centering
    \subfloat[Percentage of offloading users.\label{Fig:Percentage_of_OffloadingUsers_vs_Weights_Comparison}]{\includegraphics[width=0.48\linewidth]{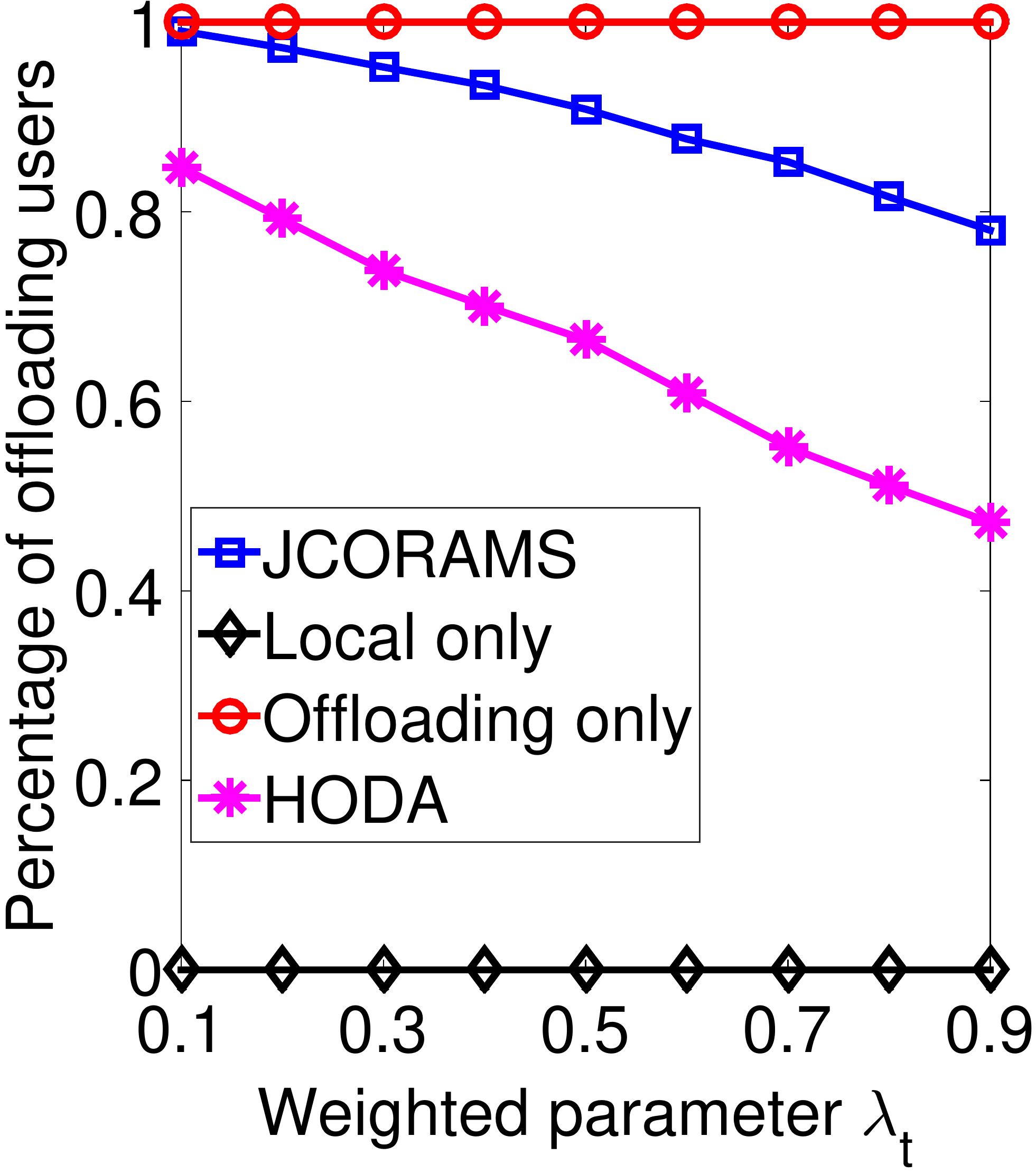}}\;
    \subfloat[Total computation overhead.\label{Fig:Computation_Overhead_vs_Weights_Comparison}]{\includegraphics[width=0.49\linewidth]{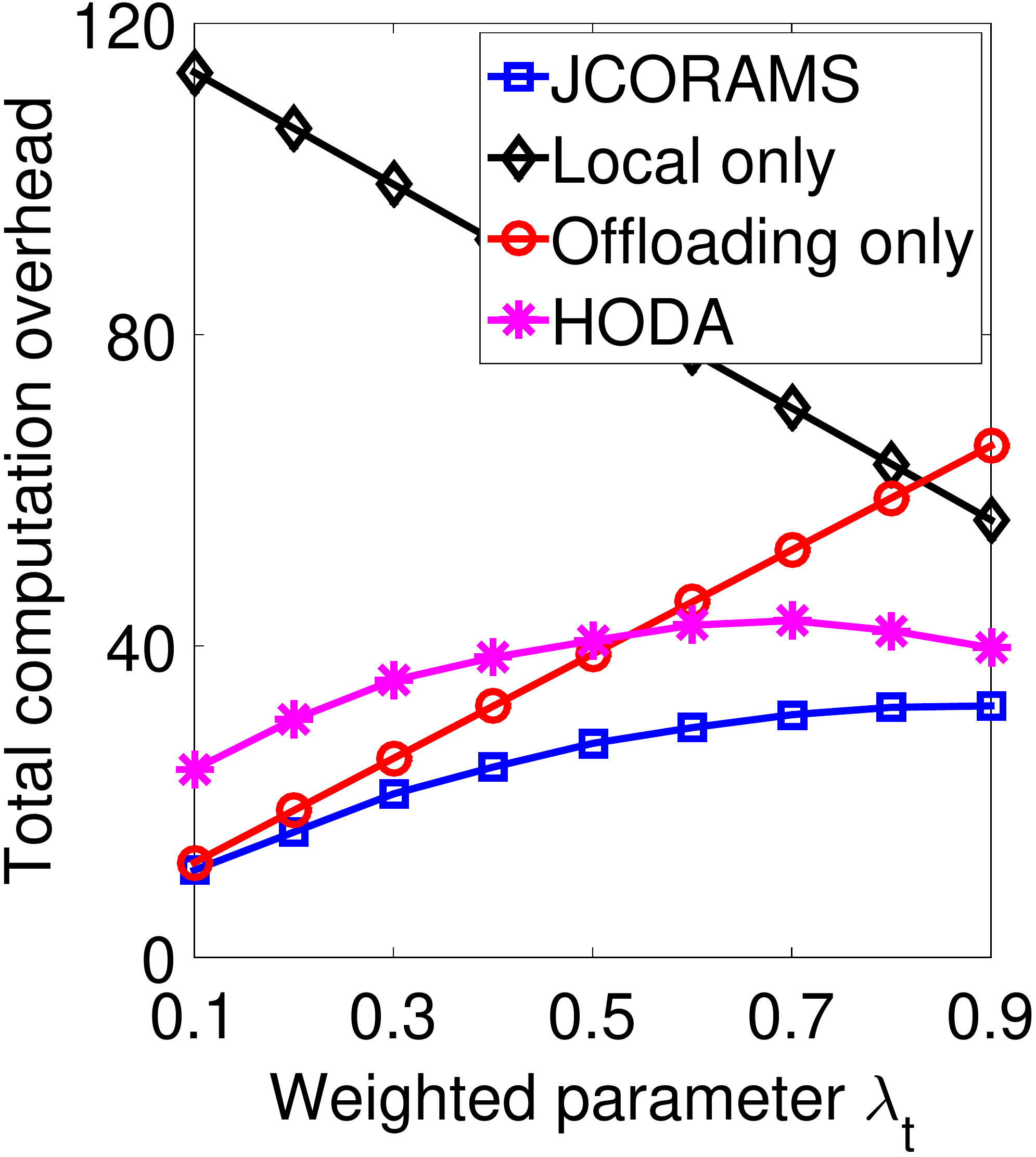}}
  \caption{Comparison of our proposed approach and three alternative frameworks under different weighted parameters.}
  \label{Fig:PerformanceComparison_vs_Weights}
\end{figure}

In the fifth experiment, we discuss the impacts of the maximum transmit power on the performances of the considered approaches. From the Fig.~\ref{Fig:Performance_vs_maxTransmitPower}, it is seen that when the maximum transmit power $p_{n}^{\max}$ increases, the percentage of offloading users and system-wide computation overhead increases and decreases, respectively, and all become saturated when $p_{n}^{\max}$ is sufficiently large. For example, with $N = 36$ and $p_{n}^{\max} = 0.55$ (W) for all mobile users, the percentage of offloading users is $100\%$, i.e., the performances of our proposed and the offloading only schemes are the same, and the computation overhead is $17.9$. This is due to the fact that increasing the offloading rates makes the time and energy costs for offloading the computation tasks smaller and as a consequence, there are more mobile users that tend to offload their computation tasks to the MEC servers. 
\begin{figure}[!ht]
  \centering
    \subfloat[Percentage of offloading users.\label{Fig:Percentage_of_OffloadingUsers_vs_maxTransmitPower_Comparison}]{\includegraphics[width=0.48\linewidth]{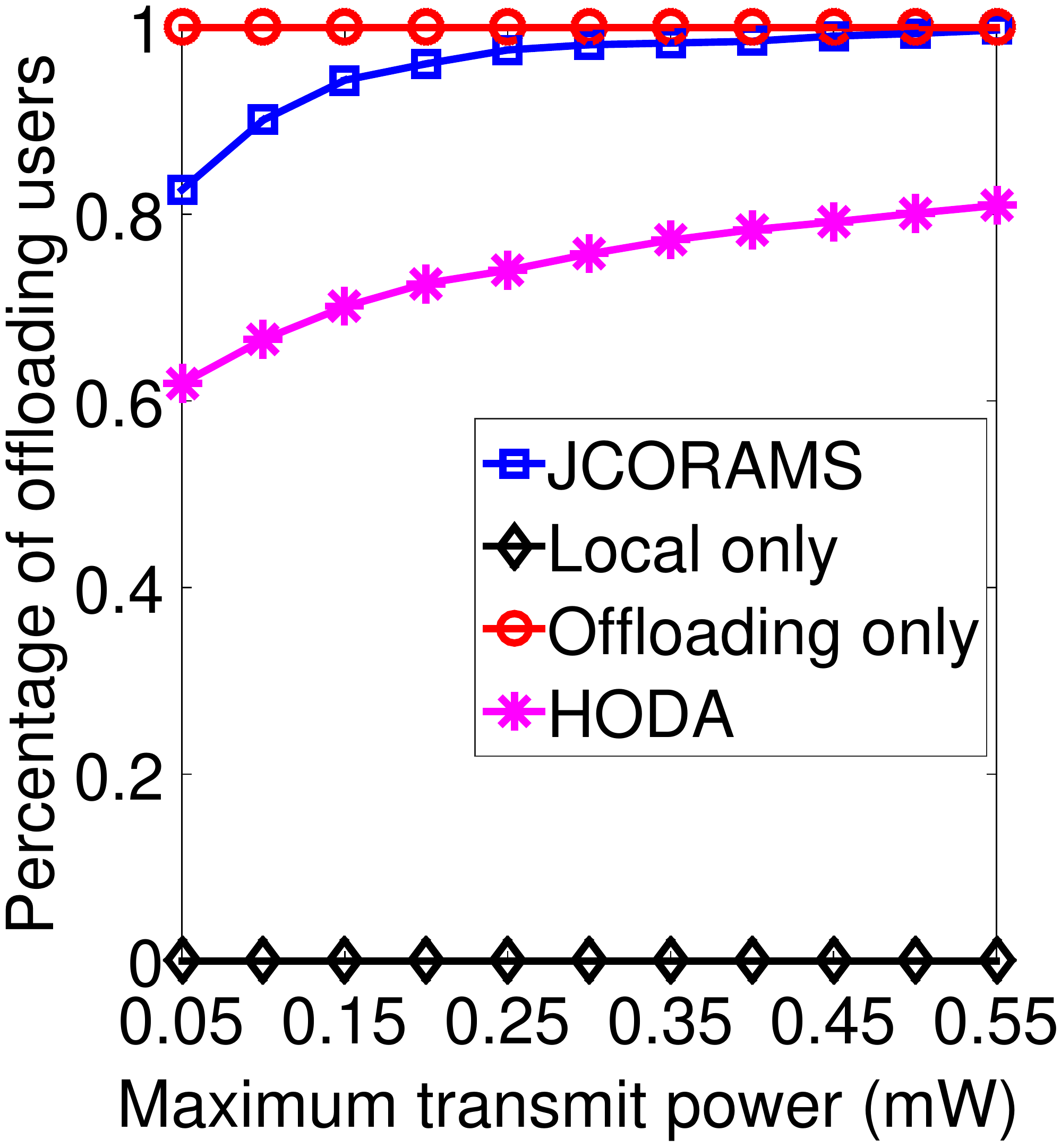}}\;\;
    \subfloat[Total computation overhead.\label{Fig:Computation_Overhead_vs_maxTransmitPower_Comparison}]{\includegraphics[width=0.48\linewidth]{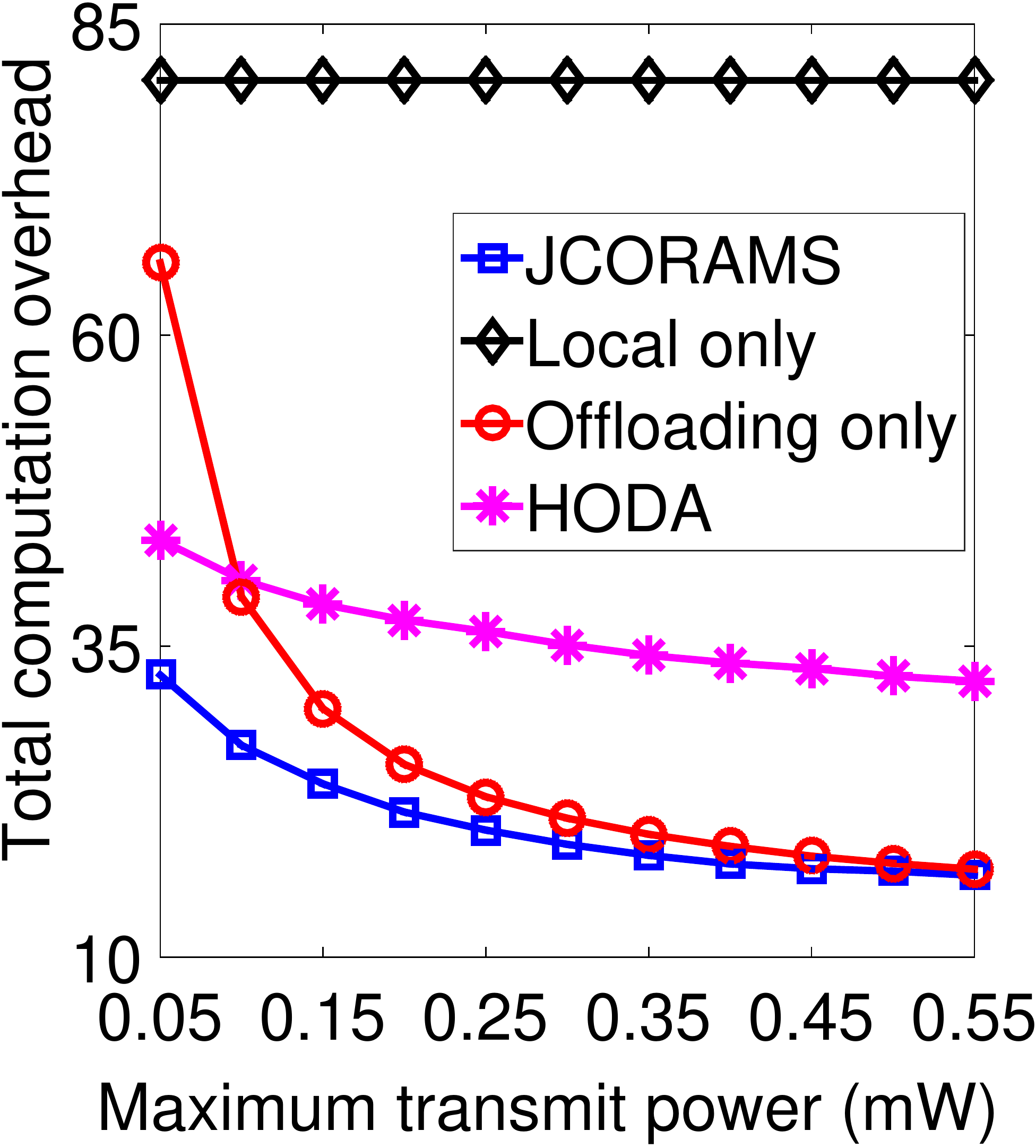}}
  \caption{Comparison of our proposed approach and three existing frameworks under variant maximum transmit power.}
  \label{Fig:Performance_vs_maxTransmitPower}
\end{figure} 

The final experiment represents the number of offloading users and system-wide computation overhead for the different algorithms when the maximum computational capability $ f_{0} $ of MEC servers varies from $0.5$ GHz to $4$ GHz. It is shown in Fig.~\ref{Fig:Percentage_of_OffloadingUsers_vs_CC_f0} that the percentage of offloading users monotonically increases with the computational capability of the MEC servers, and the increasing rate gradually decreases, i.e., increasing the computational capability of the MEC servers from $0.5$ GHz to $1.0$ GHz makes more users benefit from computation offloading than that from $1.0$ GHz to $1.5$ GHz. The reason is that when the computational capability of MEC servers is small, the execution time is high and so the remote computation overhead becomes higher than the local computation overhead. At the same time, because more mobile users benefit from computation offloading, the system-wide computation overhead decreases. In order to evaluate the optimality of the proposed algorithm, we compare JCORAMS with the hJTORA algorithm proposed in \cite{Tuyen2017Joint}, where at each step, MEC server and subchannel selections are heuristically found by solving all possible resource allocation problems, and the algorithm ends when there is no feasible way to increase the objective value. Observe from Fig.~\ref{Fig:Computation_Overhead_vs_CC_f0} that at the maximum computational capability $ f_{0} = 3.5 $ GHz the proposed algorithm generates the total computation overhead of $ 63.579 $, which is close to that of hJTORA with the gap of $ 8.86\% $. Fig.~\ref{Fig:Computation_Overhead_vs_CC_f0} also depicts that when the inter-cell interference is taken into consideration, offloading all computation tasks to the MEC servers is very inefficient. This is due to the fact that i) the locations of mobile users and MEC servers are randomly distributed in each simulation realization, so some mobile users may have very bad connections to the MEC servers, and ii) when the inter-cell interference exists and become severe, the offloading rates of offloading users are relatively low and the offloading time becomes much higher. In this scenario mobile users, with the bad connections and severe interference, should locally handle their computations, while the other send requests to the MEC servers for computation offloading. As a result, a joint optimization of offloading decision, resource allocation, and interference management is highly needed to improve the network performance, which is clearly demonstrated by the comparison between our proposed algorithm and hJTORA  with the local and offloading only schemes in Fig.~\ref{Fig:Performance_vs_CC_f0}.
\begin{figure}[!ht]
	\centering
	\subfloat[Percentage of offloading users.\label{Fig:Percentage_of_OffloadingUsers_vs_CC_f0}]{\includegraphics[width=0.475\linewidth]{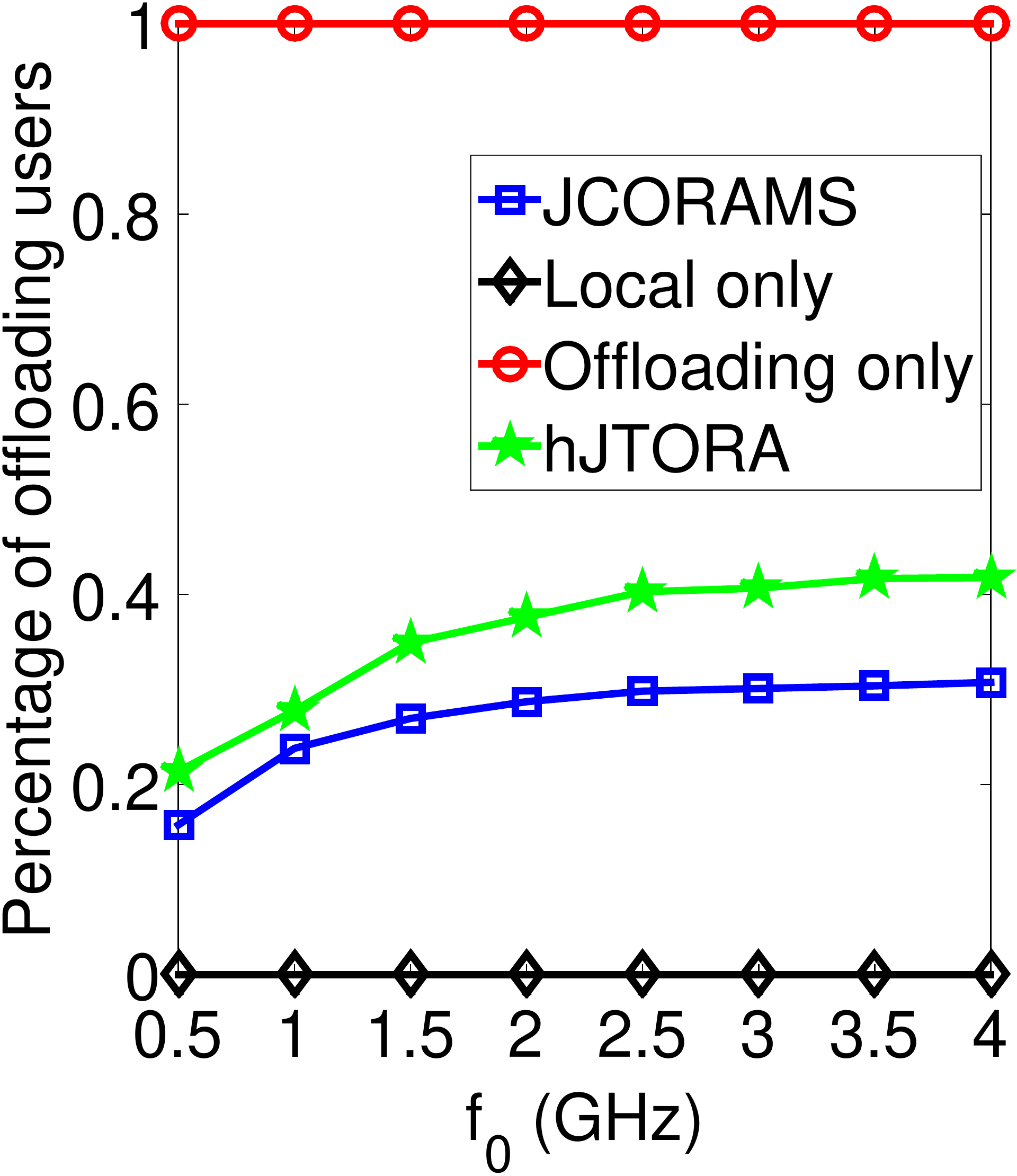}}\;\;
	\subfloat[Total computation overhead.\label{Fig:Computation_Overhead_vs_CC_f0}]{\includegraphics[width=0.46\linewidth]{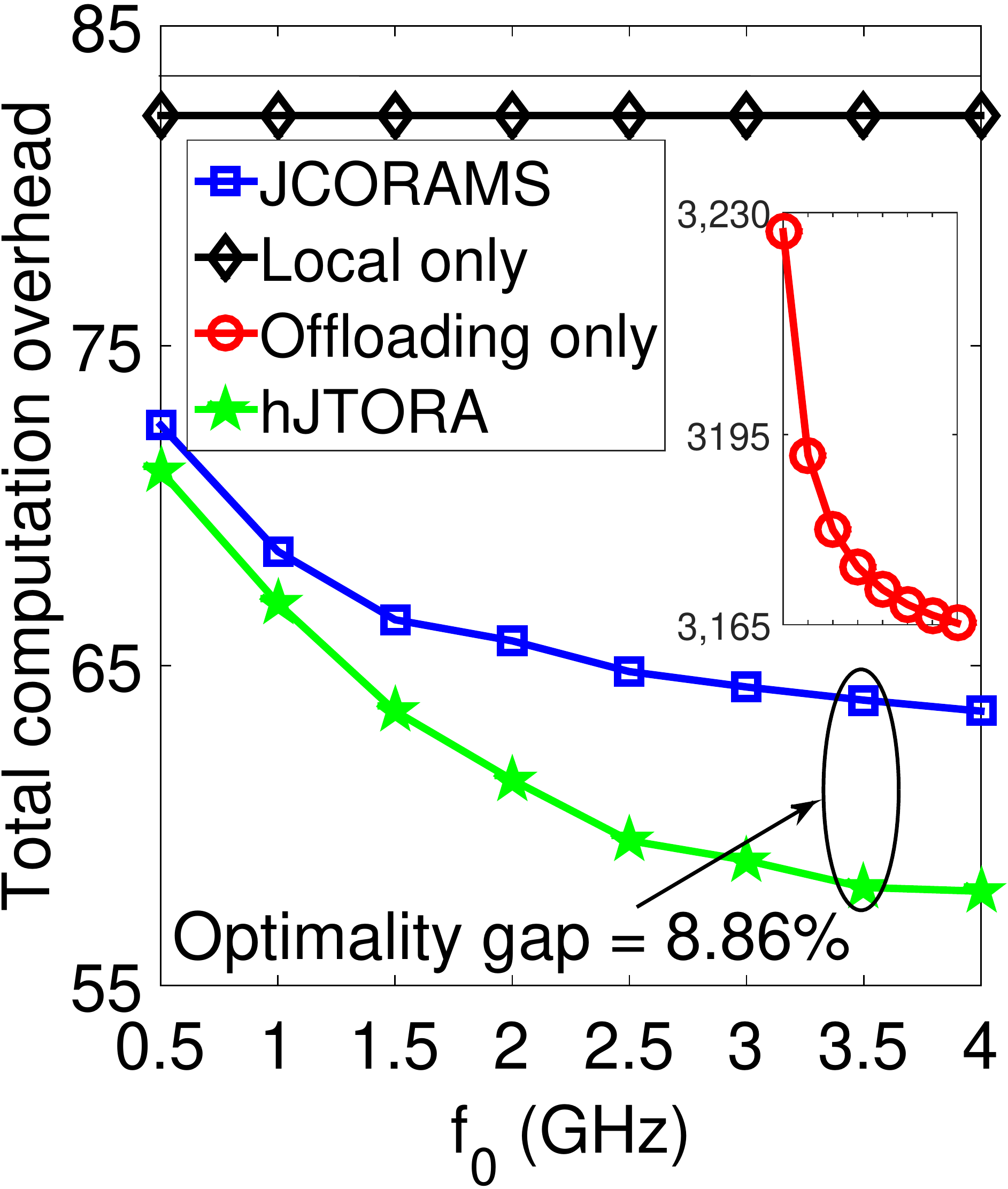}}
	\caption{Performance of the proposed algorithm under different computational capability of MEC servers.}
	\label{Fig:Performance_vs_CC_f0}
\end{figure}

\section{Conclusion and Future Work}
\label{Sec:Conclusion}
In this paper, we proposed an optimization problem for jointly determining the computation offloading decision and allocating the transmit power of mobile users and computation resources at the MEC servers. Our proposed framework is different from existing ones in that 1) we consider HetNets with multiple MEC servers and 2) propose a decentralized computation offloading scheme. The simulation results validated that the proposed algorithm can achieve better performances than alternative frameworks.

A joint framework of resource allocation and server selection in collocation edge computing systems is currently under investigation of our ongoing work. Moreover, we will take into account the effects of computation offloading to the quality of service of macrocell users. Finally, we will consider of hierarchical MEC systems for differentiated applications of mobile users where users with latency-sensitive applications offload their tasks to the first tier at small cells while users with latency-tolerant applications offload their tasks to the second MEC server tier at macrocells. 




\end{document}